\documentclass[11pt]{article}
\pdfoutput=1 
\usepackage{amssymb}
\usepackage{amsmath}
\usepackage{amstext}
\usepackage{graphicx}
\usepackage{epsfig}
\usepackage{verbatim} 
\usepackage{fancybox}
\usepackage{color}
\usepackage{mathdots}
\usepackage{ulem}
\usepackage{enumitem}
\usepackage{caption}
\usepackage{amsthm}
\usepackage{subfig}
\usepackage{mathtools}
\usepackage{youngtab}
\usepackage{bbm}
\usepackage{parskip}
\usepackage[numbers,sort&compress]{natbib}
\usepackage{ytableau}
\usepackage[utf8]{inputenc}
\usepackage{tikz}
\usepackage{simplewick}
\usetikzlibrary{positioning, trees,decorations, decorations.markings, decorations.pathmorphing, arrows, shapes.geometric}

\newcommand{\Comment}[1]{{}}
\definecolor{darkblue}{rgb}{0.15,0.35,0.55}
\definecolor{reddish}{rgb}{0.65, 0.2, 0.2}
\usepackage[linktocpage=true]{hyperref}
\hypersetup{
colorlinks=true,
citecolor=darkblue,
linkcolor=reddish,
urlcolor=darkblue,
pdfauthor={},
pdftitle={},
pdfsubject={}
}

\usepackage{cleveref}

\newcommand{\be}{\begin{equation}}
\newcommand{\ee}{\end{equation}}
\newcommand{\nn}{\nonumber}
\newcommand{\intM}{\int_{\mathcal{M}} }

\numberwithin{equation}{section}

\makeatletter
\def\thickhline{%
  \noalign{\ifnum0=`}\fi\hrule \@height \thickarrayrulewidth \futurelet
   \reserved@a\@xthickhline}
\def\@xthickhline{\ifx\reserved@a\thickhline
               \vskip\doublerulesep
               \vskip-\thickarrayrulewidth
             \fi
      \ifnum0=`{\fi}}
\makeatother

\newlength{\thickarrayrulewidth}
\newtheorem*{theorem}{Theorem}

\begin{document}

\renewcommand{\thefootnote}{\fnsymbol{footnote}}
~
\vspace{1.75truecm}
\thispagestyle{empty}
\begin{center}
{\LARGE \bf{
Bootstrap Bounds on Closed Hyperbolic Manifolds 
}}
\end{center} 

\vspace{1cm}
\centerline{\Large James Bonifacio\footnote{\href{mailto:james.j.bonifacio@gmail.com}{\texttt{james.j.bonifacio@gmail.com}}}
}
\vspace{.5cm}

\centerline{{\it Department of Applied Mathematics and Theoretical Physics,}}
 \centerline{{\it   University of Cambridge, Wilberforce Road, Cambridge CB3 0WA, U.K.}} 
 \vspace{.25cm}
 
\vspace{1cm}
\begin{abstract}
\noindent
The eigenvalues of the Laplace--Beltrami operator and the integrals of products of eigenfunctions must satisfy certain consistency conditions on compact Riemannian manifolds. These consistency conditions are derived by using spectral decompositions to write quadruple overlap integrals in terms of products of triple overlap integrals in multiple ways. In this paper, we show how these consistency conditions imply bounds on the Laplacian eigenvalues and triple overlap integrals of closed hyperbolic manifolds, in analogy to the conformal bootstrap bounds on conformal field theories. We find an upper bound on the gap between two consecutive nonzero eigenvalues of the Laplace--Beltrami operator in terms of the smaller eigenvalue, an upper bound on the smallest eigenvalue of the rough Laplacian on symmetric, transverse-traceless, rank-2 tensors, and bounds on integrals of products of eigenfunctions and eigentensors. Our strongest bounds involve  numerically solving semidefinite programs and are presented as exclusion plots. We also prove the analytic bound $\lambda_{i+1} \leq 1/2+3 \lambda_i+\sqrt{\lambda_i^2+2 \lambda_i+1/4}$  for consecutive nonzero eigenvalues of the Laplace--Beltrami operator on closed orientable hyperbolic surfaces. We give examples of genus-2 surfaces that nearly saturate some of these bounds. 
To derive the consistency conditions, we make use of a transverse-traceless decomposition for symmetric tensors of arbitrary rank.
\end{abstract}

\newpage

\setcounter{tocdepth}{2}
\tableofcontents
\renewcommand*{\thefootnote}{\arabic{footnote}}
\setcounter{footnote}{0}

%\newpage

\section{Introduction}
Hyperbolic manifolds are important in many areas of mathematics and physics, including number theory, low-dimensional topology, dynamical systems, and string theory. A hyperbolic manifold $(\mathcal{M}, \hat{g})$ is a Riemannian manifold of constant sectional curvature $-1$, so its Riemann curvature tensor can be written in terms of its metric as
\be \label{eq:constant-curvature}
R_{mnpq} = \kappa \left( \hat{g}_{mp} \hat{g}_{nq} - \hat{g}_{mq} \hat{g}_{np} \right),
\ee
where $\kappa=-1$. 
In this paper, we study bootstrap bounds on Laplacian eigenvalues and integrals of products of eigenfunctions on hyperbolic manifolds that are closed, i.e., compact without boundary.\footnote{The manifolds in this work are always assumed to be smooth, connected, and orientable.}
Bootstrap bounds are bounds that arise from certain consistency conditions that eigenvalues and overlap integrals must satisfy, which are analogous to the consistency conditions used in the conformal bootstrap for conformal field theories (CFTs)  \cite{Rattazzi:2008pe, Rychkov:2009ij, Caracciolo:2009bx, Poland:2011ey, Kos:2014bka, Poland:2018epd}. Such bounds were studied for general closed Einstein manifolds in Ref.~\cite{Bonifacio:2020xoc}. By restricting to hyperbolic manifolds, we can obtain additional consistency conditions and stronger bounds.

To give a simple example of the consistency conditions we study, consider the Laplace--Beltrami operator $\Delta$ on  $(\mathcal{M}, \hat{g})$ and its eigenfunctions $\phi_i$ with eigenvalues $\lambda_i$,
\be
\Delta \phi_i = \lambda_i \phi_i, \quad i \in \mathbb{Z}_{\geq 0},
\ee
where we take the eigenfunctions to be real and orthonormal. The eigenvalues are ordered by increasing magnitude and with multiplicity,
$
\lambda_0=0< \lambda_1 \leq \lambda_2 \leq \dots \rightarrow \infty.
$
The eigenfunctions form a basis for the space of square-integrable functions on $(\mathcal{M}, \hat{g})$, so we can expand the product of any two eigenfunctions as a sum over eigenfunctions,
\be \label{eq:eigenExpansion}
\phi_{i} \phi_{j} = \sum _{k=0}^{\infty} c_{i j k} \phi_k.
\ee
The coefficients in this spectral decomposition are the integrals of products of three eigenfunctions,
\be
c_{i j k} \coloneqq \int_{\mathcal{M}} dV  \,\phi_i \phi_j \phi_k ,
\ee
where $d V $ is the Riemannian volume form. Given the integral of any product of eigenfunctions, we can write it in terms of the triple overlap integrals $c_{ijk}$ by repeatedly using the eigenfunction expansion \eqref{eq:eigenExpansion}. For example, we can expand a quadruple overlap integral as
\be
\int_{\mathcal{M}} dV  \phi_i^2 \phi_j^2 = \sum_{k=0}^{\infty} c_{i j k}^2,
\ee
which is just an instance of Parseval's identity.
However, this is not the only way to write this integral in terms of triple overlap integrals. By expanding different pairs of eigenfunctions, we get
\be
\int_{\mathcal{M}} dV  \phi_{i}^2 \phi_{j}^2 = \sum_{k=0}^{\infty} c_{i i k} c_{j j k} .
\ee
Since these are two ways of evaluating the same integral, we obtain the consistency condition
\be 
\sum_{k=0}^{\infty} \left( c_{i j k}^2 -c_{i i k} c_{j j k}\right) =0,
\ee
which holds on any closed Riemannian manifold.

By considering more complicated integrals involving derivatives of eigenfunctions and using spectral decompositions for tensor fields, we can obtain additional consistency conditions involving Laplacian eigenvalues and integrals of products of eigenfunctions and eigentensors. In this paper, we consider consistency conditions coming from quadruple overlap integrals involving a single fixed scalar eigenfunction. There are infinitely many consistency conditions of this form on closed hyperbolic manifolds, whereas there are finitely many  for general closed Einstein manifolds \cite{Bonifacio:2019ioc, Bonifacio:2020xoc}. We consider integrals with up to 16 derivatives. 

Given a set of consistency conditions, we can derive bounds on eigenvalues and triple overlap integrals by following the methods of the conformal bootstrap \cite{Rattazzi:2008pe, Rychkov:2009ij, Caracciolo:2009bx, Poland:2011ey, Kos:2014bka, Poland:2018epd}, as explored in Ref.~\cite{Bonifacio:2020xoc}. An example of the type of bound we can obtain is shown in Fig.~\ref{fig:eigenvalue-bounds-0}. This plot shows an upper bound on $\lambda_2$ in terms of $\lambda_1$ for closed hyperbolic surfaces, together with the eigenvalues from a particular family of genus-2 surfaces that nearly saturate the bound. This bound is derived using numerical semidefinite programming methods without rigorous error estimates.
We also prove the following analytic version of this bound,  which is slightly weaker and shown by the dashed line in Fig.~\ref{fig:eigenvalue-bounds-0}:
\be \label{eq:analytic-intro}
\lambda_{2} \leq 1/2+3 \lambda_1+\sqrt{\lambda_1^2+2 \lambda_1+1/4}.
\ee
Both this analytic bound and the numerical one apply more generally to any pair of consecutive nonzero eigenvalues.
We also present numerical upper bounds on the following quantities:
\begin{enumerate}
\item Eigenvalue gaps for closed hyperbolic manifolds with more than two dimensions.
\item The smallest eigenvalue of the rough Laplacian on symmetric, transverse-traceless, rank-2 tensors.
\item The magnitude of the triple overlap integral of the lightest non-constant eigenfunction, normalised by the volume $V$ of the manifold.
\item The magnitudes of triple overlap integrals involving the lightest non-constant eigenfunction and holomorphic $s$-differentials in two dimensions, for $s=2$ and $s=4$, normalised by $V$.
\end{enumerate}
Our numerical bounds are presented as exclusion plots, and the data points used to produce these plots are available in an ancillary notebook. We expect that these numerical bounds could be made mathematically rigorous with additional effort.

\begin{figure}[h!t]
\begin{center}
\epsfig{file=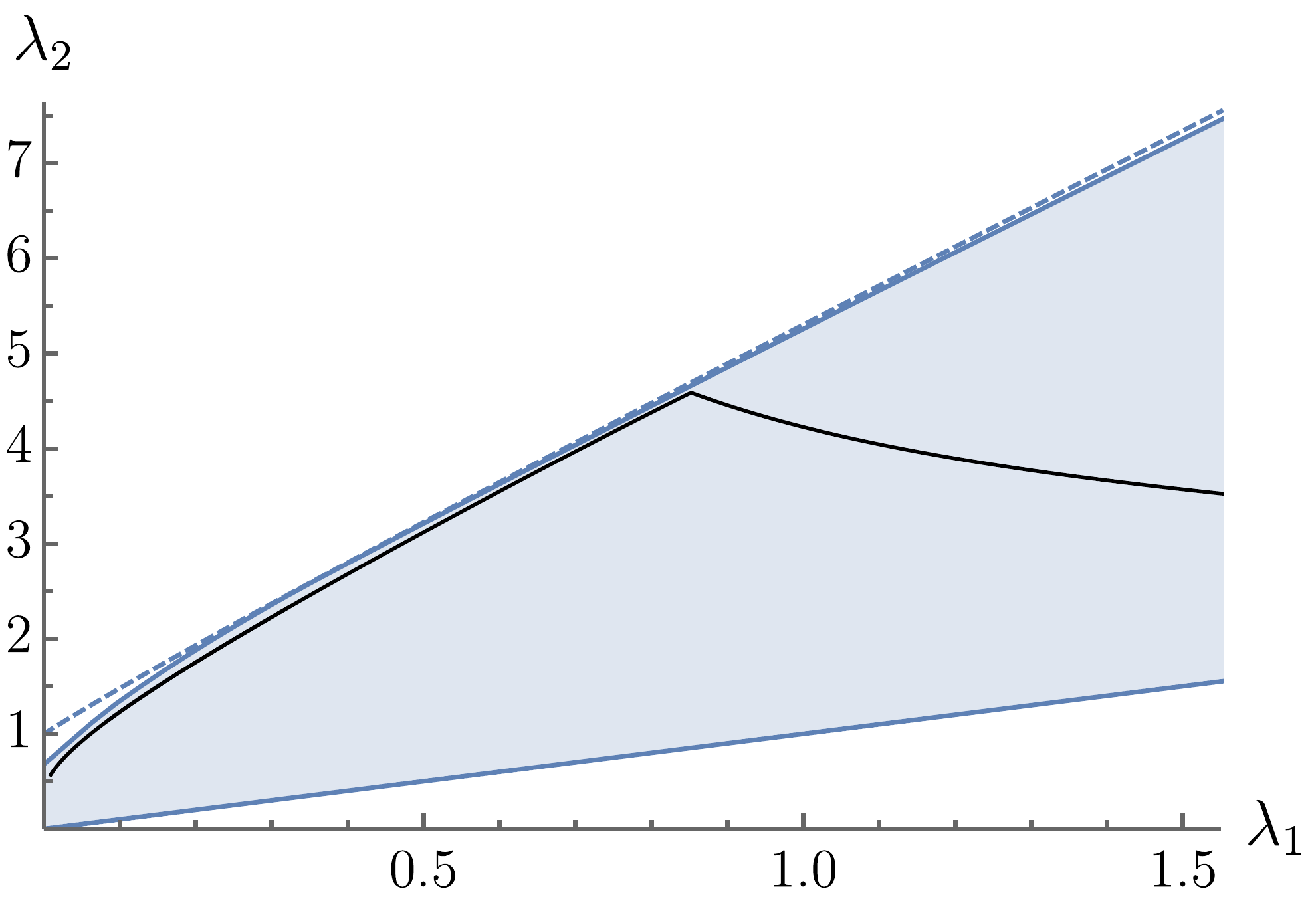, width=10cm}
\caption{An upper bound on the second positive eigenvalue of the Laplace--Beltrami operator on closed hyperbolic surfaces in terms of the first positive eigenvalue. The blue region is allowed by the numerical bound, where the lower bound is $\lambda_2 \geq \lambda_1$. The dashed line is the analytic bound \eqref{eq:analytic-intro}. The black line corresponds to the eigenvalues of the family of genus-2 surfaces $\rho(\ell)$, defined in Eq.~\eqref{eq:moduli-path}, for $\ell \in[15/4,  20]$, which were computed using the program \texttt{Hypermodes} by  Strohmaier and Uski \cite{hypermodes, Strohmaier_2012}. The same bounds apply to any pair of consecutive nonzero eigenvalues of the Laplace--Beltrami operator on a  closed hyperbolic surface.}
\label{fig:eigenvalue-bounds-0}
\end{center}
\end{figure}

The outline of the rest of this paper is as follows: in Section \ref{sec:spectralIntro}, we review some results about closed hyperbolic manifolds and their spectral theory. In Section \ref{sec:ccs}, we discuss the decomposition of symmetric tensors into symmetric, transverse-traceless eigentensors of the rough Laplacian and explain how to derive consistency conditions. In Section \ref{sec:eigenvalue-bounds}, we present bounds on eigenvalues and in Section \ref{sec:overlap-bounds} we present bounds on triple overlap integrals. We conclude in Section \ref{sec:discussion}. In Appendix \ref{app:ttDecomposition}, we prove the existence of a transverse-traceless decomposition for smooth, symmetric tensor fields of arbitrary rank on a closed Riemannian manifold, following Refs.~\cite{Pestov-Sharafutdinov, Dairbekov2011, York1973, York1974}.

\section{Spectral theory of closed hyperbolic manifolds}
\label{sec:spectralIntro}
In this section, we briefly review some results about closed hyperbolic manifolds and their spectral theory, focusing on two and three dimensions and the eigenvalues of the Laplace--Beltrami operator.
Further details and background material can be found in Refs.~\cite{buser1992, Besse, marden_2016}.

\subsection{Hyperbolic surfaces}
A closed orientable surface is classified topologically by its genus $g \in \mathbb{Z}_{\geq0}$. Surfaces with $g \geq 2$ admit hyperbolic metrics and, by the uniformisation theorem, any metric on such a surface is conformally equivalent to a unique hyperbolic metric. There is also a one-to-one correspondence between hyperbolic structures and complex structures, so we can think of a closed hyperbolic surface as a closed Riemann surface and vice versa. Every closed hyperbolic surface can be written as a quotient $ \mathbb{H}^2/\Gamma$, where $\mathbb{H}^2$ is the upper-half plane model of two-dimensional hyperbolic space and $\Gamma$ is a freely acting discrete subgroup of ${\rm PSL}(2, \mathbb{R})$, the group of orientation-preserving isometries of $\mathbb{H}^2$. A hyperbolic surface of genus $g$ has volume $4 \pi (g-1)$ by the Gauss--Bonnet theorem.

The moduli space of hyperbolic structures on a genus-$g$ surface is the space of hyperbolic metrics on the surface modulo diffeomorphisms. Its universal cover is a
$(6g-6)$-dimensional space called Teichm\"uller space, which is the space of hyperbolic metrics on the surface modulo diffeomorphisms that are isotopic to the identity.
Moduli space and Teichm\"uller space are related by the mapping class group.
Every hyperbolic surface of genus $g$ can be obtained by glueing together $2g-2$ hyperbolic pairs of pants along their $3g-3$ pairs of geodesic boundaries. Each such gluing is completely specified by two real numbers, $\ell>0$ and $\tau$, describing the boundary lengths and their relative twisting---see Fig.~\ref{fig:genus-2}. We can therefore specify a hyperbolic surface by the $6g-6$ parameters $\ell_i$, $\tau_i$, $i=1, \dots, 3g-3$ of such a pants decomposition, although such a decomposition is not unique. These parameters are called Fenchel--Nielsen coordinates, and they give a parametrisation of Teichm\"uller space. 

There are many results concerning the spectrum of the Laplace--Beltrami operator on closed orientable hyperbolic surfaces. The eigenvalue 1/4 plays a special role in this subject---it is the bottom of the spectrum on $\mathbb{H}^2$ and eigenvalues smaller than 1/4 are called small.
The number of small eigenvalues is bounded above in terms of the genus $g$ since $\lambda_{2g-2}>1/4$ \cite{Otal-2009}. On the other hand, for any $g\geq 2$ and $\epsilon>0$, there are genus-$g$ surfaces with arbitrarily many eigenvalues smaller than $1/4+ \epsilon$ \cite{Buser1977}. 
The spectral gap $\lambda_1$ is bounded above by a constant that tends to $1/4$ as $g \rightarrow \infty $ \cite{Huber1974}. There is no lower bound on $\lambda_1$, and in fact $\lambda_{2g-3}$ can be made arbitrarily small \cite{Buser1977}. 
The Yang--Yau bound for the first three nonzero eigenvalues gives \cite{YangYau1980}
\be
\frac{1}{\lambda_1}+\frac{1}{\lambda_2}+\frac{1}{\lambda_3} \geq \frac{3(g-1)}{2 \lfloor (g+3)/2 \rfloor},
\ee
where $\lfloor \cdot \rfloor$ is the floor function and we have used an improvement pointed out in Ref.~\cite{ElSoufi83}. This implies an upper bound on $\lambda_1$,
\be \label{eq:YangYau-1}
\lambda_1 \leq \frac{ 2 \lfloor (g+3)/2 \rfloor}{g-1} \leq 4,
\ee
and also an upper bound on $\lambda_2$ for sufficiently large $\lambda_1$, 
\be  \label{eq:YangYau-2}
\lambda_1 > \frac{2\lfloor(g+3)/2 \rfloor}{3(g-1)} \implies \lambda_2 \leq \left[ \frac{3(g-1)}{4 \lfloor(g+3)/2 \rfloor}- \frac{1}{2\lambda_1}\right]^{-1}.
\ee
For a hyperbolic surface with $g=2$, the largest value of $\lambda_1$ is conjectured to be that of the Bolza surface \cite{Strohmaier_2012}, which has $\lambda_1 \approx 3.839$.
This surface is called the Hadamard--Gutzwiller model in the quantum chaos literature \cite{Aurich1989}.
An improved upper bound on $\lambda_1$ for $g= 3$ was found recently \cite{ros2021}, namely $\lambda_1 \leq 2(4-\sqrt{7})$, and improved upper bounds on $\lambda_1$ for almost all other genera were found in Ref.~\cite{karpukhin2021}. 
A generalisation of Eq.~\eqref{eq:YangYau-1} for general eigenvalues was given in Ref.~\cite{karpukhin2020}, 
\be
\lambda_k \leq  \frac{ 2 k\lfloor (g+3)/2 \rfloor}{g-1}.
\ee

There are also interesting recent results concerning the size of $\lambda_1$ on generic closed hyperbolic surfaces of genus $g$ as $g \rightarrow \infty$. It was shown that generically $\lambda_1 > 3/16- \epsilon$  for any $\epsilon>0$ as $g \rightarrow \infty$ \cite{lipnowski2021, wu2021}, where generic here means with probability tending to one using the normalised Weil--Petersson measure on moduli space. This builds on Mirzakhani's result that generically $\lambda_1 > 0.0024$ as $g \rightarrow \infty$ \cite{Mirzakhani2013}. It is conjectured that this $3/16$ can be replaced by $1/4$ \cite{Wright2020}. There is a similar result for the eigenvalues of random covers of a surface \cite{magee2020}.
 
\begin{figure}
\centering
\hspace*{.1cm}{\resizebox{8cm}{!}{%
\begin{tikzpicture}
\node at (0,0) {\includegraphics[width=15cm]{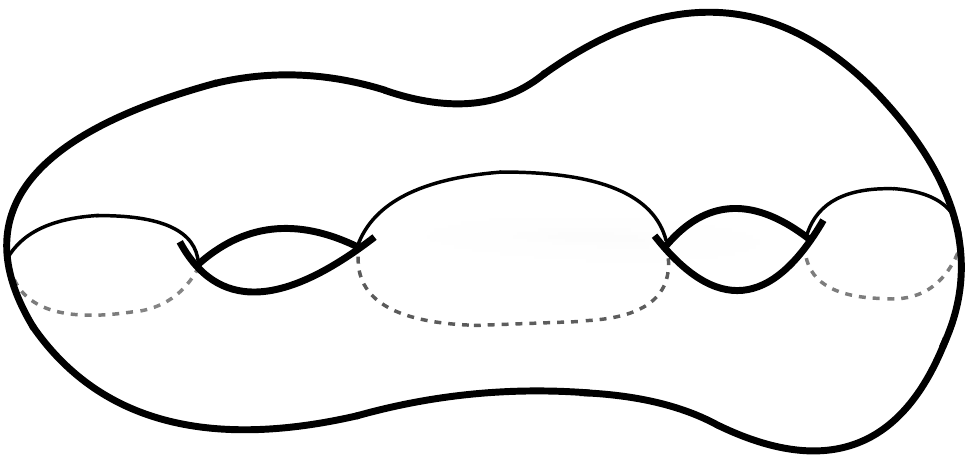}};
\node[scale=1.8] at (-5.7,.75) {$\ell_1, \, \tau_1$};
\node[scale=1.8] at (0.4,1.4) {$\ell_2, \, \tau_2$};
\node[scale=1.8] at (5.9,1.15) {$\ell_3, \, \tau_3$};
\end{tikzpicture}
}}
\caption{Fenchel--Nielson coordinates for a genus-2 surface.}
\label{fig:genus-2}
\end{figure}
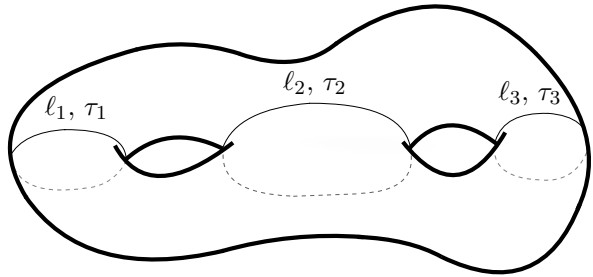

\subsection{Hyperbolic 3-manifolds}
The classification of closed 3-manifolds was pioneered by Thurston with his geometrisation conjecture \cite{Thurston82}, which was proved by Perelman \cite{Perelman2002}. Of the eight geometries appearing in the geometrisation conjecture, hyperbolic manifolds are the most complex and the only ones yet to be fully classified. 
Unlike in two dimensions, a closed hyperbolic 3-manifold has a unique hyperbolic metric. In fact, closed hyperbolic $d$-manifolds with isomorphic fundamental groups are isometric if $d \geq 3$ by Mostow rigidity \cite{Mostow68}. 
However, as in two dimensions, there is also a sense in which most 3-manifolds are hyperbolic. For example, any orientable closed 3-manifold can be obtained from $S^3$ by performing Dehn surgery along some link $L$, i.e., by excising a neighbourhood around $L$ and filling in a solid torus for each component of $L$, such that $S^3 \setminus L$ admits a complete hyperbolic metric, and most manifolds obtained from $S^3$ by performing Dehn surgery along $L$ admit a hyperbolic metric \cite{Thurston82}. 
Another important result is the virtual fibering conjecture proved by Agol \cite{Agol2013}, which implies that every closed hyperbolic 3-manifold has a finite cover that is a surface bundle over $S^1$.

We mention just a few spectral results; see, e.g., Ref.~\cite{Callahan-thesis} for a more detailed review.  Schoen showed that on a closed hyperbolic 3-manifold $\lambda_1 \geq \pi^2 /(3\times2^{34}V)^2$ \cite{Schoen1982}. This means that $\lambda_1$ cannot be  arbitrarily small for a fixed volume, although it can be arbitrarily small  \cite{Callahan-thesis}.
Explicit numerical results for the low-lying eigenvalues of certain closed hyperbolic 3-manifolds can be found in the cosmology literature \cite{Inoue:1998nz, Cornish1999, Inoue2001}; this includes the Weeks manifold, which has $\lambda_1 \approx 27.8$ and $\lambda_2 \approx 32.9$ \cite{Cornish1999}. The Weeks manifold has volume  $V \approx 0.9427$, the smallest of any closed orientable hyperbolic 3-manifold \cite{Gabai2007}, and can be obtained by performing Dehn surgery on the Whitehead link. 

\section{Tensor decompositions and consistency conditions}
\label{sec:ccs}

In this section, we explain how to derive consistency conditions for closed hyperbolic manifolds. The derivation requires introducing a transverse-traceless spectral decomposition for symmetric tensors of arbitrary rank. 

\subsection{Symmetric tensors}
Let $T^{(s)}$ and $T'^{(s)}$ be real, symmetric, rank-$s$ tensor fields on a closed $d$-dimensional Riemannian manifold $(\mathcal{M}, \hat{g})$. The tensor fields in this paper, which we refer to as tensors for simplicity, are always assumed to be smooth. The canonical global inner product on the space of such tensors is the following $L^2$ product:
\be \label{eq:innerProduct}
(T^{(s)}, T'^{(s)}) \coloneqq \intM dV \hat{g}^{m_1 n_1} \dots \hat{g}^{m_s n_s} T^{(s)}_{m_1 \dots m_s} T'^{(s)}_{n_1 \dots n_s}.
\ee
This inner product induces an $L^2$ norm on symmetric tensors.

\subsection{Transverse-traceless decomposition}
We want to decompose a symmetric tensor into symmetric, transverse-traceless tensors. 
The first step is to decompose a symmetric tensor $T^{(s)}$ into a symmetric, transverse-traceless tensor plus a trace term and the symmetrised traceless derivative of a lower-rank, symmetric, traceless tensor,
\be  \label{eq:ttDecompose}
T^{(s)}_{m_1 \dots m_s} = T^{(s)TT}_{m_1 \dots m_s} + \nabla_{(m_1} W^{(s-1)}_{m_2 \dots m_s)}- \frac{s-1}{2s+d-4}\hat{g}_{(m_1 m_2} \nabla^n W^{(s-1)}_{m_3 \dots m_s) n}+\hat{g}_{(m_1 m_2} \bar{T}^{(s-2)}_{m_3 \dots m_s)},
\ee
where $\nabla$ is the covariant derivative, $T^{(s)TT}$ is transverse and traceless, $ W^{(s-1)}$ is traceless, and $\bar{T}^{(s-2)}$ can be written in terms of traces of $T^{(s)}$. Note that we symmetrise with weight one.
This decomposition is possible on any closed Riemannian manifold, as discussed in Appendix \ref{app:ttDecomposition}. On a manifold of constant curvature, we can iteratively apply this decomposition to the trace and longitudinal pieces to get an orthogonal decomposition of an arbitrary symmetric tensor in terms of symmetric, transverse-traceless tensors,
\be
\label{eq:TensorDecomposition}
T^{(s)}_{m_1 \dots m_s} = \sum_{t=0}^s \sum_{a=0}^{\lfloor\frac{s-t}{2}\rfloor} \hat{g}_{(m_{1}  m_{2}} \dots \hat{g}_{m_{2a -1} m_{2a}}\left( \nabla_{m_{2a+1}} \cdots \nabla_{m_{s-t}} h^{(t)}_{a, m_{s-t+1} \dots m_s)}+ \dots \right),
\ee
where $h^{(t)}_{a}$ are symmetric, transverse-traceless, rank-$t$ tensors and the expression in parenthesis is fixed by adding terms to make it symmetric and traceless. More explicitly, the expression in parenthesis is $P_{s-2a-1} P_{s-2a-2} \dots P_t \, h^{(t)}_a$, where $P_n$ is the symmetrised traceless derivative defined in Eq.~\eqref{eq:PsDef}. 

As some examples of the decomposition \eqref{eq:TensorDecomposition}, for $s=1, 2, 3$ on a hyperbolic manifold we get
\begin{align}
T^{(1)}_{m_1} & = h^{(1)}_{0, m_1}+\nabla_{m_{1}} h^{(0)}_0 , \\
T^{(2)}_{m_1 m_2} & =h^{(2)}_{0, m_1 m_2}+\nabla_{(m_1} h^{(1)}_{0, m_2)} + \left(\nabla_{(m_1} \nabla_{m_2)} + \frac{1}{d}\hat{g}_{m_1m_2} \Delta\right) h^{(0)}_0 +  \hat{g}_{m_1 m_2} h^{(0)}_1, \\
T^{(3)}_{m_1 m_2 m_3} & =h^{(3)}_{0, m_1 m_2 m_3} + \nabla_{(m_1} h^{(2)}_{0, m_2 m_3)}  +\left( \nabla_{(m_1} \nabla_{m_2} h^{(1)}_{0, m_3)} +\frac{\Delta +d-1}{d+2}\hat{g}_{(m_1 m_2} h^{(1)}_{0, m_3)}\right)  + \hat{g}_{(m_1 m_2} h^{(1)}_{1, m_3)} \nonumber \\
&+  \left( \nabla_{(m_1} \nabla_{m_2} \nabla_{m_3)} h^{(0)}_0 + \frac{3 \Delta +2(d-1)}{d+2}\hat{g}_{(m_1 m_2} \nabla_{m_3)} h^{(0)}_0 \right)+\hat{g}_{(m_1 m_2} \nabla_{m_3)} h^{(0)}_1,
\end{align}
where $\Delta$ is the rough Laplacian defined below.
For $s=1$, this is the Hodge decomposition of a 1-form, which exists on any closed Riemannian manifold. For $s=2$, this is the transverse-traceless decomposition of a symmetric, rank-$2$ tensor \cite{York1973, York1974}, combined with the 1-form Hodge decomposition, which exists as an orthogonal decomposition on closed Einstein manifolds. It is useful for studying general Kaluza--Klein reductions of gravity \cite{Hinterbichler_2014}.

\subsection{Laplacians and eigenmodes}
From now on we assume that $(\mathcal{M}, \hat{g})$ is a closed hyperbolic manifold, unless otherwise stated. 
Consider the rough Laplacian on a symmetric, rank-$s$ tensor $T^{(s)}$ for any $s\in \mathbb{Z}_{\geq 0}$, defined as minus the divergence of the gradient,\footnote{For some applications, a more natural operator to consider is the Lichnerowicz Laplacian $\Delta_L$, which on a manifold of constant curvature $\kappa$ is given by
\be
\Delta_L T^{(s)}_{m_1 \dots m_s}= \Delta T^{(s)}_{m_1 \dots m_s} + \kappa s(s+d-2)T^{(s)}_{m_1 \dots m_s} -\kappa s(s-1) \hat{g}_{(m_1 m_2} T^{(s)}_{m_3 \dots m_{s}) n}{}^n.
\ee
The Lichnerowicz Laplacian commutes with several other standard operators defined on symmetric tensors. 
For symmetric, traceless tensors on constant curvature backgrounds, the Lichnerowicz Laplacian and rough Laplacian differ only by a constant, so their symmetric, traceless eigentensors are the same. 
}
\be
\Delta T^{(s)}_{m_1 \dots m_s}   \coloneqq - \nabla^n \nabla_{n} T^{(s)}_{m_1 \dots m_s}.
\ee
This operator is non-negative, essentially self-adjoint, and strongly elliptic. 
By standard results of spectral theory, the spectrum of $\Delta$ is discrete with a possible accumulation point only at infinity, each eigenspace is finite-dimensional, and the eigentensors are smooth and form a basis for the space of square-integrable, symmetric tensors. 

The rough Laplacian on functions is just the Laplace--Beltrami operator. As in the introduction, we denote its eigenvalues in non-decreasing order by $\lambda_i$ for $i\in \mathbb{Z}_{\geq 0}$, 
\be
0=\lambda_0 < \lambda_1 \leq \lambda_2\leq \lambda_3 \leq \dots,
\ee
where $\lambda_i \rightarrow \infty$ as $i \rightarrow \infty$ and eigenvalues are repeated according to their multiplicities. The associated real orthonormal eigenfunctions are denoted by $\phi^{(0)}_i$ or $\phi_{i}$. The unique zero mode is $\phi_0 = V^{-1/2}$, which we usually treat separately from the other eigenfunctions. 

Given a symmetric eigentensor of the rough Laplacian, we can decompose it into transverse-traceless tensors using Eq.~\eqref{eq:TensorDecomposition}. Using the constant curvature condition, each term in the decomposition separately satisfies an eigenvalue equation since the different terms remain orthogonal after acting with $\Delta$.
This means that we can restrict to transverse-traceless eigentensors without loss of generality. We denote the eigenvalues of the rough Laplacian on symmetric, transverse-traceless tensors of rank $s \geq 1$ for $d>2$ in non-decreasing order by $\lambda^{(s)}_i$ for $i\in \mathbb{Z}_{>0}$,
\be
\lambda^{(s)}_1 \leq \lambda^{(s)}_2\leq \lambda^{(s)}_3 \leq \dots,
\ee
where $\lambda^{(s)}_i \rightarrow \infty$ as $i \rightarrow \infty$ and eigenvalues are repeated according to multiplicity. We denote the associated real, orthonormal, symmetric, transverse-traceless, rank-$s$ eigentensors by $\phi^{(s)}_i$ for $i\in \mathbb{Z}_{>0}$,
\begin{align}
\Delta \phi^{(s)}_{i, m_1 \dots m_s}  = \lambda_{i}^{(s)} \phi^{(s)}_{i, m_1 \dots m_s}, \quad \nabla^{m_1} \phi^{(s)}_{i, m_1 \dots m_s}  = \hat{g}^{m_1 m_2} \phi^{(s)}_{i, m_1 \dots m_s} =0,  \quad  \left( \phi^{(s)}_{i},\phi^{(s)}_{j}\right) = \delta_{i j} \,.
\end{align}
On a two-dimensional manifold, there are only finitely many symmetric, transverse-traceless, rank-$s$ eigentensors for $s\geq2$, as discussed below; for these cases, we use the same notation for eigenvalues and eigentensors, but with the index running over a finite range.

Using the completeness of the eigentensors, we can expand any square-integrable, symmetric, transverse-traceless tensor in terms of $\phi_i^{(s)}$. Performing such an expansion on each $h^{(t)}_{a}$ in Eq.~\eqref{eq:TensorDecomposition}, we obtain the decomposition of a square-integrable, symmetric tensor in terms of symmetric, transverse-traceless eigentensors of the rough Laplacian,
\begin{align}
\label{eq:TensorEigenDecomposition}
T^{(s)}_{m_1 \dots m_s} & = \sum_{t=0}^s \sum_{a=0}^{\lfloor\frac{s-t}{2}\rfloor}  \sum_{i=1}^{\infty}C^{(t)}_{ a,i} \hat{g}_{(m_{1}  m_{2}} \dots \hat{g}_{m_{2a-1} m_{2a}}\left( \nabla_{m_{2a+1}} \cdots \nabla_{m_{s-t}}\phi^{(t)}_{i , m_{s-t+1} \dots m_s)}+ \dots \right) \nn \\
&+ C^{(0)}_{ s/2,0} \hat{g}_{(m_{1}  m_{2}} \dots \hat{g}_{m_{s-1} m_{s})}\phi_{0},
\end{align}
where the second line is only present when $s$ is even and $C^{(t)}_{a,i}$ are constants that can be determined by taking appropriate inner products of each side of this equation. The term in parenthesis is fixed by adding terms to make it symmetric and traceless, as in Eq.~\eqref{eq:TensorDecomposition}. For $d=2$, the index $i$ is only summed over a finite range for $t\geq 2$.

\subsection{Eigenvalue lower bounds}
The eigenvalues of the Laplace--Beltrami operator are non-negative, $\lambda_i \geq 0$. For $s >0$, we can obtain a nonzero lower bound on the eigenvalues of the rough Laplacian by considering the positivity of the following norm:
\be
\int_{\mathcal{M}} dV \left( \nabla_{[m_1} \phi^{(s)}_{i, m_2] m_3 \dots m_{s+1}} \right)^2 \geq 0,
\ee
where $\phi^{(s)}_{i}$ is a transverse-traceless eigentensor of the rough Laplacian with eigenvalue $\lambda_i^{(s)}$.
Evaluating this using integration by parts, we get the following lower bound, given in Ref.~\cite{Dyatlov2015}:
\be \label{eq:lowerbound}
\lambda_i^{(s)} \geq s+d-2, \quad s \geq 1.
\ee
This is analogous to the CFT unitary bound for symmetric, traceless tensors. The bound is saturated by symmetric, transverse-traceless eigentensors satisfying $\nabla_{[m_1} \phi^{(s)}_{i, m_2] m_3 \dots m_{s+1}}=0$. For $s=1$, these are harmonic 1-forms. If $(\mathcal{M}, \hat{g})$ can be immersed as a minimal hypersurface in a round sphere, then its second fundamental form is a rank-2 tensor satisfying this condition \cite{Simons68}.

\subsubsection{Holomorphic \texorpdfstring{$s$}{TEXT}%
-differentials}

When $d=2$, we can say more about higher-rank, transverse-traceless tensors. First we show that their eigenvalues must saturate the above lower bound, following an argument from Ref.~\cite{higuchi87}. Let $\phi^{(s)}_{i}$ be a symmetric, transverse-traceless eigentensor with $s\geq1$ in two dimensions. We can write 
\be \label{eq:2dtrick}
\nabla_{[m_1} \phi^{(s)}_{i, m_2] m_3 \dots m_{s+1}} = \epsilon_{m_1 m_2} \theta^{(s-1)}_{i, m_3 \dots m_{s+1}},
\ee
where $ \theta^{(s-1)}_{i}$ is a symmetric tensor and $\epsilon$ is the constant antisymmetric tensor.
For $s\geq2$, contracting the left-hand side with $\hat{g}^{m_2 m_3}$ gives zero since $\phi^{(s)}_{i}$ is transverse and traceless. This implies that $  \theta^{(s-1)}_{i}$ must vanish and so $\phi^{(s)}_{i}$ satisfies the condition for saturating the lower bound \eqref{eq:lowerbound}. Therefore, on a closed  hyperbolic surface the only transverse-traceless, rank-$s$ tensors with $s\geq 2$ are those with $\lambda_i^{(s)}= s$. For $s=1$, taking the divergence of Eq.~\eqref{eq:2dtrick} gives
\be
(\lambda_i^{(1)} -1) \phi^{(1)}_{i, m_1} = -2\epsilon_{m_1 m_2} \nabla^{m_2} \theta^{(0)}_{i}.
\ee
This shows that a transverse vector either has $\lambda_i^{(1)} = 1$, in which case it is a harmonic 1-form, or it has $\lambda_i^{(1)} >1$, in which case it can be written as the dual of the gradient of a non-constant scalar eigenfunction $\phi_j$ and $\lambda_i^{(1)}= \lambda_j+1$.

To better understand these transverse-traceless tensors in $d=2$, it is helpful to take a complex-analytic point of view (some useful physics references for this are Refs.~\cite{Alvarez:1982zi, dHokerPhong1986, dHokerPhong88}). In a neighbourhood of every point in two dimensions, there exist local isothermal coordinates in which the metric takes the form 
\be
ds^2 = e^{2 \sigma} \delta_{ij} dx^i dx^j, \quad i, j=1,2.
\ee
Defining complex coordinates $z= x^1 + ix^2$ and $\bar{z} = x^1 -i x^2$, the metric is 
\be
ds^2 = e^{2 \sigma}  dz d\bar{z} .
\ee
A symmetric, traceless, rank-$s$ tensor $T^{(s)}$ has only two independent components in two dimensions, which can be taken as
\begin{align}
T^{(s)}_{z \dots z} & = (T^{(s)}_{\bar{z} \dots \bar{z}})^* =\frac{1}{2} \left(T^{(s)}_{1 \dots 1 1} - iT^{(s)}_{1 \dots 1 2} \right), \\
T^{(s) z \dots z} & = \hat{g}^{z \bar{z}} \dots \hat{g}^{z \bar{z}} T^{(s)}_{\bar{z} \dots \bar{z}} = 2^{s-1} e^{-2 s \sigma} \left(T^{(s)}_{1 \dots 1 1} + iT^{(s)}_{1 \dots 1 2} \right),
\end{align}
which are called the negative- and positive-helicity components of $T^{(s)}$. 

A symmetric, traceless tensor in two dimensions is transverse if and only if its negative-helicity component is holomorphic, i.e., $\partial_{\bar{z}}T^{(s)}_{z \dots z} =0$, where $\partial_{\bar{z}} =(\partial_1 +i \partial_2)/2$. By the Riemann--Roch theorem, the vector space of holomorphic $s$-differentials with $s \geq 2$ on a closed surface of genus $g \geq 2$  has complex dimension $(2s-1)(g-1)$. Therefore, the vector space of symmetric, transverse-traceless, rank-$s$ tensors has real dimension $2(2s-1)(g-1)$ for $s\geq 2$.
For example,  there are $6(g-1)$ independent symmetric, transverse-traceless, rank-2 tensors, given by the real parts of holomorphic quadratic differentials, and these span the cotangent space of Teichm\"uller space. Holomorphic $s$-differentials with $s>2$ are important in higher Teichm\"uller theory.
The transverse vectors with $\lambda^{(1)}=1$ are harmonic 1-forms and correspond to the real parts of holomorphic differentials.
The vector space of holomorphic differentials on a genus-$g$ surface has complex dimension $g$, so there are $2g$ independent real transverse vectors with $\lambda^{(1)}=1$.

\subsubsection{Spherical manifolds}
\label{ssec:spherical}
Let us comment on manifolds of constant positive curvature. On a closed manifold of constant curvature $\kappa$, the norms of the symmetrised traceless derivatives of transverse-traceless eigentensors of the rough Laplacian take the form
\be \label{eq:harmonic-norms}
\int_{\mathcal{M}} dV \left(\nabla_{(m_{1}} \cdots \nabla_{m_{s'-s}} \phi^{(s)}_{i,  m_{s'-s+1} \dots m_{s'})}+ \dots \right)^2 = c_{s, s'}(d) \prod_{l=s}^{s'-1} \left[ \lambda_{i}^{(s)}- \kappa \left( l(l+d-1)-s \right)\right],
\ee
where $c_{s, s'}(d)$ is a rational function of $d$ that is positive for $d>1$. The only way these norms can all be non-negative when $\kappa >0$ is if the eigenvalues are restricted to specific discrete values, in which case the norms vanish for sufficiently large $s'$. In particular, the eigenvalues must be a subset of the eigenvalues of the round sphere,
\be
\kappa >0 \implies \lambda_{i}^{(s)} \in \{ \kappa \left( l(l+d-1)-s \right): \, l \in \mathbb{Z}, \, l \geq s\} ,
\ee
since otherwise the right-hand side of Eq.~\eqref{eq:harmonic-norms} would be negative for certain $s'$.  This also follows from the compactness of spheres and the Killing--Hopf theorem, which says that complete manifolds with constant positive curvature are quotients of the round sphere. 

\subsection{Triple overlap integrals}
It will be useful to introduce notation for certain triple overlap integrals involving eigenfunctions and eigentensors. For scalar eigenfunctions $\phi_{i}$ and $\phi_{j}$ and a  transverse-traceless, rank-$s$ eigentensor $\phi^{(s)}_{k}$, we define $c^{(s)}_{i j k} $ to be the following triple overlap integral:
\be
c^{(s)}_{i j k} \coloneqq \intM d V  \,\phi_{i} \nabla^{m_1} \dots \nabla^{m_s} \phi_{j} \phi^{(s)}_{k, m_1 \dots m_s},
\ee
which corresponds to the unique cubic interaction between two scalars and a spin-$s$ particle.
We often drop the superscript for the scalar overlap integrals, i.e.,  $c_{i j k} \coloneqq c^{(0)}_{i j k}$. In this paper, we only consider triple overlap integrals with $i=j$, which can  be non-vanishing only when $s$ is even. In two dimensions, these overlap integrals can be written in terms of the real parts of overlap integrals of eigenfunctions and holomorphic $s$-differentials,\footnote{The imaginary parts come from overlap integrals involving the constant antisymmetric tensor.}
\be \label{eq:overlap2D}
c^{(s)}_{i j k}= 2\intM d V \, \operatorname{Re} \left[ \phi_{i} \nabla^{z} \dots \nabla^{z} \phi_{j} \phi^{(s)}_{k, z \dots z} \right] ,
\ee
where $\nabla^z = \hat{g}^{z \bar{z}} \partial_{\bar{z}}$ when acting on a tensor of definite helicity.

These triple overlap integrals appear when expanding symmetrised products of derivatives of eigenfunctions using spectral decompositions, e.g.,
\be
\phi_{i}^2  = V^{-1}+\sum_{j=1}^{\infty} c_{iij} \phi_j , \quad \phi_i \partial_{m_1} \phi_i  =\frac{1}{2} \sum_{j=1}^{\infty}  c_{iij} \partial_{m_1} \phi_j \, .
\ee
Physically, we can think of these overlap integrals as cubic coupling constants since they determine the strengths of cubic interactions of Kaluza--Klein modes in theories with extra dimensions---see, for example, Refs.~\cite{Arefeva86, Kaloper:2000jb, DeLuca:2021pej} for some old and recent discussions of hyperbolic extra dimensions. 
Scalar overlap integrals on hyperbolic manifolds have also been studied in the mathematics literature due to their relation with integral representations of $L$-functions \cite{Sarnak94, Petridis95, Bernstein1999, Bernstein2006}. 

\subsection{Consistency conditions}
We can now introduce the consistency conditions that are the main tool used in this paper. Using the decomposition of a symmetric tensor in terms of transverse-traceless eigentensors of the rough Laplacian, we can reduce a quadruple overlap integral to a sum of products of triple overlap integrals. This reduction can be done in multiple ways, and equating the different expressions gives consistency conditions on the eigenvalues and triple overlap integrals. 
Fixing a single eigenfunction $\phi_i$, we consider identities of the following form:
\begin{align}
&\intM dV \contraction{}{\nabla_{m_1} \dots \nabla_{(m_{s_1}} }{\phi_i  }{\nabla_{m_{s_1+1}}\dots \nabla_{m_{s_1}+m_{s_2}}\phi_i} 
\contraction{\nabla_{m_{1}} \dots \nabla_{m_{s_1}} \phi_i \nabla_{m_{s_1+1}} \dots \nabla_{m_{s_1+s_2}}\phi_i}{\nabla_{m_{s_1+s_2+1}} \dots \nabla_{m_{s_1+s_2+s_3})}}{}{\phi_i \nabla^{(m_1} \dots \nabla^{m_{s_1+s_2+s_3})}\phi_i}
\nabla_{(m_1} \dots \nabla_{m_{s_1}} \phi_i \nabla_{m_{s_1+1}} \dots \nabla_{m_{s_1+s_2}} \phi_i \nabla_{m_{s_1+s_2+1}} \dots \nabla_{m_{s_1+s_2+s_3})}\phi_i \nabla^{(m_1} \dots \nabla^{m_{s_1+s_2+s_3})} \phi_i
\nonumber \\
=&
\intM d V \contraction{}{\nabla_{m_1} \dots \nabla_{(m_{s_1}} }{\phi_i \nabla_{m_{s_1+1}}\dots \nabla_{m_{s_1}+m_{s_2}}\phi_i }{\nabla_{m_{s_1+s_2+1}} \dots \nabla_{m_{s_1+s_2+s_3})}} 
\contraction[2ex]{\nabla_{m_{1}} \dots \nabla_{m_{s_1}} \phi_i}{ \nabla_{m_{s_1+1}} \dots \nabla_{m_{s_1+s_2}}\phi_i}{\nabla_{m_{s_1+s_2+1}} \dots \nabla_{m_{s_1+s_2+s_3})}}{\phi_i \nabla^{(m_1} \dots \nabla^{m_{s_1+s_2+s_3})}\phi_i}
\nabla_{(m_1} \dots \nabla_{m_{s_1}} \phi_i \nabla_{m_{s_1+1}} \dots \nabla_{m_{s_1+s_2}} \phi_i \nabla_{m_{s_1+s_2+1}} \dots \nabla_{m_{s_1+s_2+s_3})}\phi_i \nabla^{(m_1} \dots \nabla^{m_{s_1+s_2+s_3})} \phi_i, \label{eq:wick-integral}
\end{align}
where $s_i$ are non-negative integers and the Wick contraction notation means that we decompose the indicated pair of tensors according to Eq.~\eqref{eq:TensorEigenDecomposition}.  
Each side of this equation can then be written as a sum over products of triple overlap integrals using orthonormality. 

For fixed non-negative integers $s_1$, $s_2$, and $s_3$, we can get up to two consistency conditions from Eq.~\eqref{eq:wick-integral}. We consider all independent consistency conditions of this form with $2(s_1+s_2+s_3)\leq \Lambda$, i.e., with at most $\Lambda$ derivatives, for some fixed even integer $\Lambda$.
In this paper we take $\Lambda= 16$.
 For $d>2$, the number of independent consistency conditions for $\Lambda=2, 4, \dots , 16$ is $1, 2, 3, 5, 7, 9, 12, 15$. For $d=2$, the number of independent consistency conditions equals $\Lambda/2$ for $\Lambda \leq 16$. 

The simplest example of a consistency condition with a single fixed eigenfunction comes from a two-derivative integrand. It is given by the following relation:
\be
\intM dV \,
\contraction{}{\partial_m \phi_{i}}{}{\partial^m \phi_{i} }
\contraction{\partial_m \phi_{i} \partial^m \phi_{i}}{\phi_{i}}{}{\phi_{i}}
\partial_m \phi_{i} \partial^m \phi_{i} \phi_{i} \phi_{i} 
=
\intM dV \,
\contraction{}{\partial_m \phi_{i}}{\partial^m \phi_{i} }{ \phi_{i}}
\contraction[8pt]{\partial_m \phi_{i}}{ \partial^m \phi_{i}}{\phi_{i}}{\phi_{i}}
\partial_m \phi_{i} \partial^m \phi_{i} \phi_{i} \phi_{i}.
\ee
After evaluating the contractions, this gives the consistency condition
\be
V^{-1} \lambda_{i} + \sum_{j=1}^{\infty} \left( \lambda_{i} -\frac{3}{4} \lambda_j \right) c_{i i j}^2 = 0,
\ee
which holds on any closed manifold. This consistency condition appears as a sum rule for scattering amplitudes in Yang--Mills theories dimensionally reduced on a general closed manifold \cite{Bonifacio:2019ioc}. The first instance of this that we are aware of is in Ref.~\cite{Csaki:2003dt}.

Suppose that we have found all of the consistency conditions up to some fixed $\Lambda$. 
We can put these consistency conditions in the following form:
\be \label{eq:sumRules}
V^{-1} \vec{F}_{0}(\lambda_i, 0) + \sum_{\substack{s=0 \\ s \, \text{even}}}^{\Lambda/2} \sum_{j=1}^{\infty} \frac{  \vec{F}_{s}(\lambda_i, \lambda_j^{(s)})}{\chi_{s, \Lambda} (\lambda^{(s)}_j)} \left(c^{(s)}_{i i j}\right)^2 =0,
\ee
where, in a given dimension, $\vec{F}_{s}$ are vectors whose components are polynomials of the eigenvalues  $\lambda_i$ and $\lambda^{(s)}_{j}$. The functions $\chi_{s, \Lambda}$ are polynomials defined by\be \label{eq:chiDef}
\chi_{s, \Lambda} (\lambda^{(s)}_{j}) \coloneqq \prod_{\substack{l=s \\ l \, \, {\rm odd}}}^{\Lambda/2-1} \left[ \lambda_{j}^{(s)}+  l(l+d-1)-s\right],
\ee
which are positive for eigenvalues satisfying the lower bound \eqref{eq:lowerbound}.\footnote{
Let us explain why only odd values of $l$ occur in Eq.~\eqref{eq:chiDef}.
Consider the following class of integrals on a manifold of constant curvature $\kappa$:
\be
I_{i, j}^{s, s_1, s_2} \coloneqq \int_{\mathcal{M}} d V \left(\nabla_{m_{1}} \cdots \nabla_{m_{s_1+s_2-s}} \phi^{(s)}_{j,  m_{s_1+s_2-s+1} \dots m_{s_1+s_2})}+ \dots \right) \nabla^{(m_1} \dots \nabla^{m_{s_1}} \phi_i \nabla^{m_{s_1+1}} \dots \nabla^{m_{s_1+s_2})} \phi_i .
\ee
When deriving the consistency conditions, we must divide such integrals by the norms in Eq.~\eqref{eq:harmonic-norms}, which is how the polynomials $\chi_{s, \Lambda}$ arise when $\kappa=-1$.
By integrating by parts and commuting derivatives, we can write 
\be
I_{i, j}^{s, s_1, s_2} = p_{s,s_1,s_2} \! \left(\lambda_i, \lambda^{(s)}_{j} \right) c^{(s)}_{i i j}, 
\ee
where $p_{s,s_1,s_2}$ is a polynomial depending also on $d$ and $\kappa$.
Suppose that $\kappa>0$, so we can write $\lambda_j^{(s)} = \kappa(l(l+d-1)-s)$ for some integer $l\geq s$. It follows from the discussion in Sec.~\ref{ssec:spherical} that $I_{i, j}^{s, s_1, s_2} $ vanishes for $s_1+s_2 \geq l+1$. If $l$ is odd, then $c^{(s)}_{i i j}$ vanishes by parity when $s$ is even. If $l \leq s_1+s_2-1$ is even, then $c^{(s)}_{i i j}$ can be nonzero, so $p_{s, s_1, s_2}$ must have a zero at the corresponding eigenvalue. By continuity in $\kappa$, this implies that for general $\kappa$ we must have
\be
 p_{s,s_1,s_2} \! \left(\lambda_i, \lambda^{(s)}_{j} \right)\propto \prod_{\substack{l=s \\ l \, \, {\rm even}}}^{s_1+s_2-1} \left[ \lambda_{j}^{(s)}- \kappa \left( l(l+d-1)-s\right)\right].
\ee
These zeros in $I_{i, j}^{s, s_1, s_2}  $ at even values of $l$ cancel the zeros that would otherwise occur in $\chi_{s, \Lambda}$ when $\kappa=-1$.} On a closed hyperbolic surface, there are only finitely many transverse-traceless, rank-$s$ tensors for $s\geq 2$, so in these cases the sum over $j$ in Eq.~\eqref{eq:sumRules} has a finite range. The explicit consistency conditions for $\Lambda =16$ are included in the ancillary file. 
We have verified that the $\kappa>0$ versions of these consistency conditions are satisfied by zonal spherical harmonics on round spheres of various dimensions, which is a nontrivial check of the $\vec{F}_0$ part of the consistency conditions.

These consistency conditions are closed under restricting to quotients of $\mathcal{M}$ by subgroups of its isometry group. This means that the bounds we derive also apply to such quotients, including certain unoriented manifolds, orbifolds, and manifolds with boundaries. Closed hyperbolic manifolds cannot have continuous isometries, but they can have discrete symmetries. For example, every closed genus-2 surface has a $\mathbb{Z}_2$ symmetry (a hyperelliptic involution). Hurwitz's automorphism theorem says that the maximum order of the symmetry group of a closed orientable hyperbolic surface of genus $g$ is $84(g-1)$.

To conclude this section, we note that for $\Lambda=16$ the vector of polynomials  $\vec{F}_{0}(\lambda_i, \lambda_j)$ satisfies the following condition:
\be \label{eq:3/16zero}
\vec{F}_0\left(\frac{3(d-1)^2}{16}, \frac{(d-1)^2}{4} \right) = 0.
\ee
We do not have a good explanation for this zero, but its effect will be apparent in some of our bounds since it can mark the transition between the existence or not of solutions to the semidefinite programming problems we consider.
In two dimensions, this zero occurs at the eigenvalues $3/16$ and $1/4$. Curiously, the eigenvalue $3/16$ appears in several places in the spectral theory of hyperbolic surfaces, including Selberg's $3/16$ theorem for congruence subgroups \cite{Selberg3/16} and the recent results of Refs.~\cite{lipnowski2021, wu2021, magee2020}. The eigenvalue $(d-1)^2/4$ is the bottom of the spectrum of $\mathbb{H}^d$. 

\section{Eigenvalue upper bounds}
\label{sec:eigenvalue-bounds}
In this section, we use the consistency conditions \eqref{eq:sumRules} to find bounds on the eigenvalues of closed hyperbolic manifolds. The strategy parallels that for finding bounds on scaling dimensions of CFTs \cite{Rattazzi:2008pe, Rychkov:2009ij}. We first derive an analytic bound for surfaces and then find various numerical bounds.

\subsection{An analytic bound for surfaces}
We start by deriving an analytic eigenvalue bound for closed hyperbolic surfaces. 
\begin{theorem}
Let $\lambda_i$ denote the $i$\textsuperscript{th} positive eigenvalue of the Laplace--Beltrami operator on a closed orientable hyperbolic surface. Then the following inequality holds:
 \be \label{eq:inequality}
 \lambda_{i+1}  \leq 1/2+3 \lambda_i+\sqrt{\lambda_i^2+2 \lambda_i+1/4}.
 \ee
\end{theorem}
\begin{proof}
Fix a positive integer $i$. Consider the following two consistency conditions for closed hyperbolic surfaces:
\begin{align}
Z_1 & \coloneqq V^{-1} \lambda_{i} + \sum_{j=1}^{\infty} \left( \lambda_{i} -\frac{3}{4} \lambda_j \right) c_{ii j}^2=0, \\
Z_2 & \coloneqq \sum_{j=1}^{\infty} \lambda_j \left( \lambda_j (\lambda_j-1)-6\lambda_j \lambda_i + \lambda_i(8\lambda_i+1)\right) c_{iij}^2=0.
\end{align}
These follow from the three consistency conditions that exist for $\Lambda =6$, which were studied for general closed Einstein manifolds in Refs.~\cite{Bonifacio:2019ioc, Bonifacio:2020xoc}. We define 
\be \label{eq:x-def}
 x \coloneqq \lambda_{i+1} - \left( 1/2+3 \lambda_i+\sqrt{\lambda_i^2+2 \lambda_i+1/4} \right).
\ee
Now assume that $x>0$. We want to show that this is inconsistent with $Z_1=Z_2=0$. To reach a contradiction, we consider the combination $Z_1+ \alpha Z_2$, where $\alpha$ is defined as
\be
\alpha \coloneqq \frac{32 \lambda _{i+1} \lambda _i^3-24 \lambda _{i+1}^2 \lambda _i^2+4 \lambda _{i+1} \lambda _i^2+4 \lambda _{i+1}^3 \lambda _i-4 \lambda _{i+1}^2 \lambda _i-4 \lambda _i+3 \lambda _{i+1}}{4\lambda _{i+1} \left(8 \lambda _i^2-6 \lambda _{i+1} \lambda_i+\lambda _i+\lambda _{i+1}^2-\lambda _{i+1}\right)} .
\ee

First we establish that $\alpha >0$ by writing it in terms of $x$ and  $t \coloneqq \sqrt{4 \lambda _i \left(\lambda _i+2\right)+1}$,
\be
\alpha = \frac{10 \lambda _i+4 \lambda _i x \left(2 \lambda _i \left(2 \lambda _i+3 t+4\right)+x \left(6 \lambda _i+3 t+1\right)+t+2 x^2+1\right)+3 t+6 x+3}{4 x (t+x)
   \left(6 \lambda _i+t+2 x+1\right)},
\ee
which is manifestly positive since $\lambda_i>0$, $t>0$, and, by assumption, $x>0$.
Now let us write
\be \label{eq:analytic-step}
Z_1+ \alpha Z_2 =\frac{\lambda_{i}}{V} +\sum_{j=1}^{\infty} r_j c_{ii j}^2, 
\ee
where we have defined
\be
r_j \coloneqq  \lambda_{i} -\frac{3}{4} \lambda_j  +\alpha  \lambda_j \left( \lambda_j (\lambda_j-1)-6\lambda_j \lambda_i + \lambda_i(8\lambda_i+1)\right) .
\ee
We show that the right-hand side of Eq.~\eqref{eq:analytic-step} is positive for $x>0$, thus reaching a contradiction with $Z_1=Z_2=0$. The term $\lambda_i/V$ is positive since $\lambda_i>0$ and $V>0$. Since $c_{iij}^2 \geq 0$, it only remains to show that $r_j \geq0$ for each $j \in \mathbb{Z}_{>0}$. 

Let $j \in \mathbb{Z}_{>0}$. Consider first the case $j > i$. We define $ y_j \coloneqq \lambda_j - \lambda_{i+1}$, so that $y_j \geq 0$ for $j > i$.  We can then write $r_j$ for $j > i$ as
\begin{align}
r_j= & \alpha y_j^3 + \frac{1}{2} \alpha  \left(6 \lambda _i+3 t+6 x+1\right) y_j^2 + \frac{p_1}{4 x
   (t+x) \left(6 \lambda _i+t+2 x+1\right)}y_j  \nonumber \\
&+ \frac{1}{2} \lambda _i x \left(2 \lambda _i \left(2 \lambda _i+3 t+4\right)+x \left(6 \lambda _i+3 t+1\right)+t+2 x^2+1\right),
\end{align}
where
\begin{align}
p_1 & \coloneqq 56 \lambda _i^3+124 \lambda _i^2+38 \lambda _i+36 \lambda _i^2 t+26 \lambda_i t+60 \lambda_i t x^4+240 \lambda_i^2 t x^3+40 \lambda_i t x^3+288 \lambda_i^3 t
   x^2+216 \lambda_i^2 t x^2  \nonumber \\
&  +24 \lambda_i t x^2+18 t x^2+96 \lambda_i^4 t x+208 \lambda_i^3 t x+56 \lambda_i^2 t x+52 \lambda_i t x+12 t x+3 t+24 \lambda_i
   x^5+120 \lambda_i^2 x^4+20 \lambda_i x^4 \nonumber \\
& +352 \lambda_i^3 x^3+464 \lambda_i^2 x^3+56 \lambda_i x^3+576 \lambda_i^4 x^2+1248 \lambda_i^3 x^2+336 \lambda
   _1^2 x^2+72 \lambda_i x^2+12 x^3+12 x^2+320 \lambda_i^5 x\nonumber \\
& +800 \lambda_i^4 x+416 \lambda_i^3 x+168 \lambda_i^2 x+104 \lambda_i x+12 x+3.
\end{align}
This is manifestly non-negative since all coefficients and variables are non-negative. Similarly, we define $z_j \coloneqq \lambda_i/\lambda_j -1$, which satisfies $z_j \geq0$ for $1\leq j \leq i$. We can then write $r_j$ for $1\leq j \leq i$  as
\be
r_j= \frac{\lambda _i \left(12 \alpha  \lambda _i^2+\left(4 \alpha  \lambda _i \left(8 \lambda _i+1\right)+9\right) z_j^2+\left(4 \alpha  \lambda _i \left(10 \lambda
   _i+1\right)+6\right) z_j+4 z_j^3+1\right)}{4 \left(z_j+1\right){}^3},
\ee
which is also manifestly non-negative.
 We therefore reach a contradiction and hence $x \leq 0$.
 \end{proof}

\subsection{Numerical bounds}
Let us now generalise the above reasoning in a way that can easily incorporate many consistency conditions and can be optimised numerically on a computer, following the bootstrap logic \cite{Rattazzi:2008pe, Rychkov:2009ij, Caracciolo:2009bx, Poland:2011ey, Kos:2014bka, Poland:2018epd}.  If $n_{\Lambda}$ is the number of independent consistency conditions  in Eq.~\eqref{eq:sumRules}, we consider real linear combinations by contracting with a vector $\vec{\alpha} \in \mathbb{R}^{n_{\Lambda}}$,
\be \label{eq:sumRulesContracted}
V^{-1} \vec{\alpha} \cdot \vec{F}_{0}(\lambda_i, 0) + \sum_{\substack{s=0 \\ s \, \text{even}}}^{\Lambda/2} \sum_{j=1}^{\infty} \frac{ \vec{\alpha} \cdot \vec{F}_{s}(\lambda_i,\lambda_j^{(s)})}{\chi_{s, \Lambda} (\lambda^{(s)}_j)} \left(c^{(s)}_{i i j}\right)^2 =0.
\ee
The idea is to make an assumption about the spectrum and then try to find an $\vec{\alpha}$ such that Eq.~\eqref{eq:sumRulesContracted} gives a contradiction. If this is possible, then the  assumption about the spectrum is inconsistent. 
For example, to find an upper bound on $\lambda_2$ for a given value of $\lambda_1$, we take some candidate eigenvalues $(\lambda_1^*, \lambda_2^*)$ and try to find an $\vec{\alpha}$ satisfying the following conditions:
\begin{align}
\vec{\alpha} \cdot \vec{F}_{0}(\lambda_1^*, 0) & =1, \\
\vec{\alpha} \cdot \vec{F}_{0}(\lambda_1^*, x) & \geq 0, \quad \forall x \in \{\lambda_1^*\} \cup [\lambda_2^*, \infty), \\
\vec{\alpha} \cdot \vec{F}_{s}(\lambda_1^*, x) & \geq 0, \quad \forall x \geq s+d-2, \quad s=2, 4, \dots, \Lambda/2 \,,
\end{align} 
where the last condition implements the lower bound \eqref{eq:lowerbound}.
If such an $\vec{\alpha}$ can be found, then we have a contradiction with Eq.~\eqref{eq:sumRulesContracted} since the sum of a positive number with non-negative numbers cannot be zero, so we can conclude that $\lambda_2 < \lambda_2^*$ for $\lambda_1 = \lambda_1^*$.\footnote{We can also conclude the stronger statement that there exists some $j>1$ such that $\lambda_j < \lambda_2^*$ and $c_{11j} \neq 0$.} If no such $\vec{\alpha}$ can be found, then we conclude nothing. The best upper bound is obtained by finding the smallest $\lambda_2^*$ for which there exists such an $\vec{\alpha}$. 
If the first positive eigenvalue is degenerate, then this bound applies to the second distinct positive eigenvalue.

The problem of finding an $\vec{\alpha}$ satisfying the above conditions can be formulated as a semidefinite programming problem \cite{Poland:2011ey, Kos:2014bka}. The program \texttt{SDPB} is an arbitrary-precision semidefinite program solver that is designed to find numerical solutions to problems of precisely this form \cite{Simmons-Duffin:2015qma, Landry:2019qug}. 
For the above problem, we use \texttt{SDPB} to find the best upper bound, up to some small threshold, for a discrete set of values of $\lambda_1$, and then we interpolate between these points. The bound is presented by plotting this interpolation. We similarly find and plot numerical upper bounds for the smallest eigenvalue of the rough Laplacian on symmetric, transverse-traceless, rank-$2$ tensors. The output from \texttt{SDPB} can sometimes be used to find an exact solution and we can verify a bound analytically. However, other times the solution holds only up to a small threshold specified by the user (e.g., a nonzero duality gap), which can mean that a polynomial becomes negative for very large values of its argument. For such cases, we do not know how to bound the error in terms of the user-specified parameters of \texttt{SDPB}, although any error can in principle be made arbitrarily small. In practice,
we expect the error to be much smaller than the precision to which we quote the bounds, which can be checked by recomputing bounds using higher precision and smaller thresholds.

\subsubsection{Hyperbolic surfaces}
We first consider closed hyperbolic surfaces. The numerical and analytic upper bounds on $\lambda_2$ are shown in Fig.~\ref{fig:eigenvalue-bounds-0} above.
It is also interesting to compare the numerical bootstrap bound to some existing eigenvalue bounds by plotting the allowed values of $\lambda_1$ and $\lambda_2$ for surfaces of small genus. For surfaces of genus $g\in\{2,3,4\}$, Fig.~\ref{fig:genus-2-3-4-regions} shows the regions carved out by  the numerical bootstrap bounds and the bounds of Refs.~\cite{YangYau1980,ros2021,Otal-2009, karpukhin2020}, which were reviewed in Sec.~\ref{sec:spectralIntro}. The filled circles in Fig.~\ref{fig:genus-2-3-4-regions} correspond to the surfaces conjectured to maximise the first eigenvalue for these genera \cite{Strohmaier_2012, Cook-thesis}, namely the Bolza surface, the Klein quartic, and Bring's surface.
In both Figs.~\ref{fig:eigenvalue-bounds-0} and \ref{fig:genus-2-3-4-regions}, the black lines correspond to the eigenvalues of a particular family of genus-2 surfaces described in the next paragraph. The proximity of some of these examples to the numerical bootstrap bound shows that it is nearly optimal in these regions, which is a good indication that bootstrap methods work well for hyperbolic manifolds. These genus-2 eigenvalues were computed using the \texttt{Fortran} program \texttt{Hypermodes} by Strohmaier and Uski \cite{hypermodes}, which implements an algorithm for computing eigenvalues and eigenfunctions based on the method of particular solutions \cite{Strohmaier_2012} (see Ref.~\cite{strohmaier2016} for some introductory lectures).

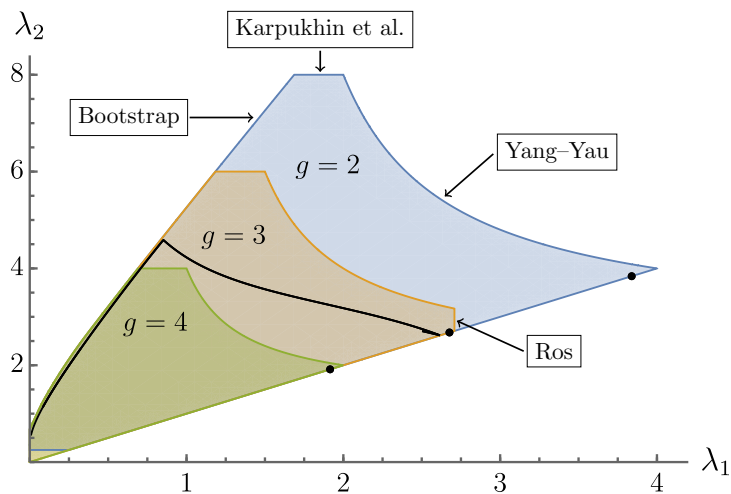
\begin{figure}
\centering
\hspace*{.1cm}{\resizebox{10cm}{!}{%
\begin{tikzpicture}
\node at (0,0) {\includegraphics[width=12cm]{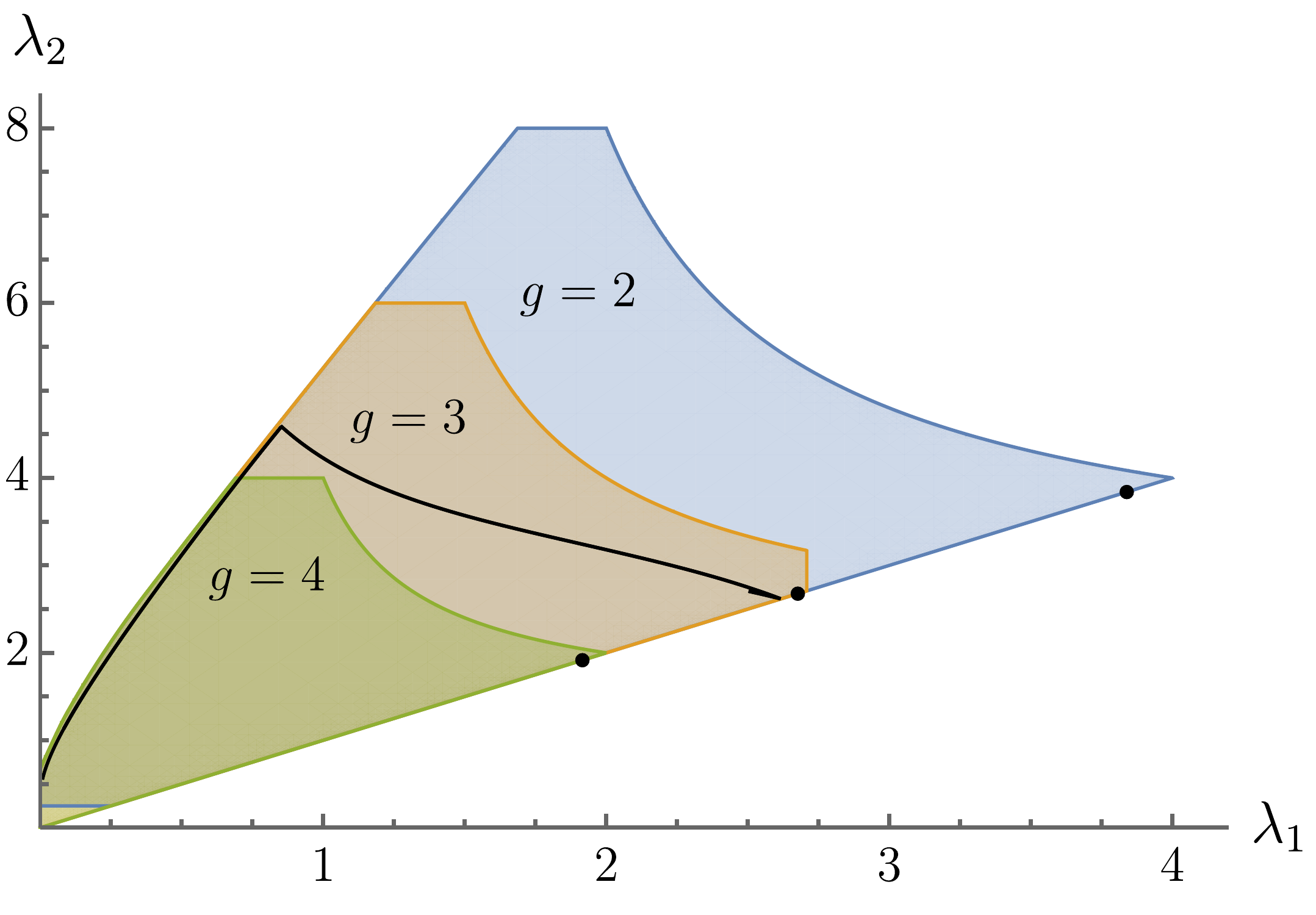}};
\draw[thick,->] (2,1.7) -- (1.2,1);
\node[draw] at (3.05,1.73) {Yang--Yau};
\draw[thick,->] (2.5,-1.55) -- (1.4,-1.05);
\node[draw] at (3,-1.55) {Ros};
\draw[thick,<-] (-.85,3.05) -- (-.85,3.4);
\node[draw] at (-.85,3.75) {Karpukhin et al.};
\draw[thick,->] (-3,2.3) -- (-1.95,2.3);
\node[draw] at (-4,2.3) {Bootstrap};
\end{tikzpicture}
}}
\caption{Allowed regions for the first two nonzero eigenvalues of the Laplace--Beltrami operator on closed hyperbolic surfaces of genus 2, 3, and 4. The regions are shaped by the bounds of Yang--Yau \cite{YangYau1980} (as given in Eqs.~\eqref{eq:YangYau-1} and \eqref{eq:YangYau-2}), Otal--Rosas \cite{Otal-2009}, Ros \cite{ros2021}, and Karpukhin et al. \cite{karpukhin2020}, together with the numerical bootstrap bound. The regions shrink as $g$ increases, except for the triangular region with $\lambda_2 \leq 1/4$, which is excluded by the Otal--Rosas bound only for $g=2$. The filled circles going from right to left correspond to the Bolza surface, the Klein quartic, and Bring's surface, which are the closed hyperbolic surfaces of genus 2, 3, and 4 with the largest symmetry groups. These surfaces are conjectured to have the largest values of $\lambda_1$ for their genera \cite{Strohmaier_2012, Cook-thesis}. The black line corresponds to the eigenvalues of the genus-2 surfaces along the path $\rho(\ell)$, defined in Eq.~\eqref{eq:moduli-path}, for $\ell \in[2,20]$, computed using \texttt{Hypermodes} \cite{hypermodes, Strohmaier_2012}.}
\label{fig:genus-2-3-4-regions}
\end{figure}

Let us now describe the family of genus-2 surfaces whose eigenvalues are shown in Figs.~\ref{fig:eigenvalue-bounds-0} and \ref{fig:genus-2-3-4-regions}. We write the Fenchel--Nielson coordinates of a genus-2 surface as $ \left(\ell_1, \tau_1; \ell_2,  \tau_2; \ell_3, \tau_3 \right) $, using the same conventions as Ref.~\cite{Strohmaier_2012}---see Fig.~\ref{fig:genus-2}. A full Dehn twist is $\tau_i =\pm 1$. Consider then the path through the Teichm\"uller space of closed genus-2 surfaces defined by 
\be \label{eq:moduli-path}
\rho: \mathbb{R}_{>0} \rightarrow (\mathbb{R}_{>0} \times \mathbb{R})^3, \quad \ell \mapsto  \left(\ell, 0; 2, \frac{1}{2}; 2, \frac{1}{2} \right).
\ee
The black lines in Figs.~\ref{fig:eigenvalue-bounds-0} and \ref{fig:genus-2-3-4-regions} correspond to the eigenvalues of surfaces along finite segments of this path. To give a clearer picture of the spectra of these surfaces, in Fig.~\ref{fig:eigenvalue-paths} we show the low-lying eigenvalues of $\rho(\ell)$ for $\ell \in[2, 8.5]$. Along real analytic paths through Teichm\"uller space, such as $\rho$, the eigenvalues are described by real analytic functions \cite{buser1992}. The curves of a given colour in Fig.~\ref{fig:eigenvalue-paths} are not analytic due to level crossings, but the underlying analytic curves can be discerned. We can thus determine the $\mathbb{Z}_2$ parity of the corresponding eigenfunctions by knowing their parity at a particular value of $\ell$. Similarly, we can check that no low-lying eigenvalues have been missed by verifying at one point along the curve that none have been missed, e.g., using one of the methods described in Ref.~\cite{strohmaier2016}, as we have done. We picked the path $\rho$ because it is easy to describe and it falls close to the bootstrap bounds, but it is possible to find genus-2 surfaces that are slightly closer to saturating the bound, even in a small neighbourhood of this path. 
\begin{figure}[h!t]
\begin{center}
\epsfig{file=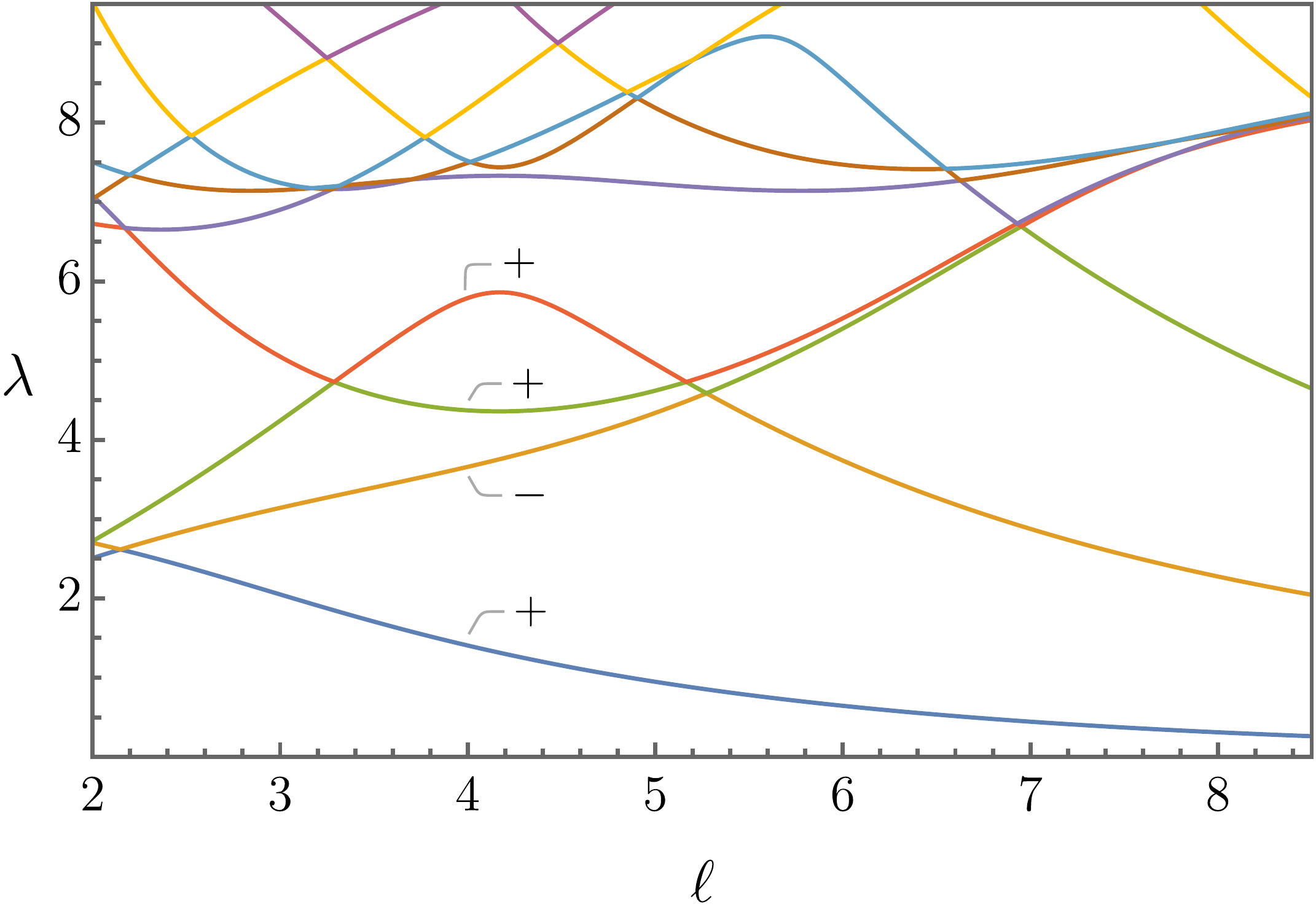, width=10cm}
\caption{Low-lying eigenvalues of the Laplace--Beltrami operator for the closed genus-2 surfaces along the path $\rho(\ell)$ in Teichm\"uller space, defined in Eq.~\eqref{eq:moduli-path}, for $\ell \in[2, 8.5]$, computed using \texttt{Hypermodes} \cite{hypermodes, Strohmaier_2012}. The $\pm$ labels indicate the $\mathbb{Z}_2$ parity of the first four nontrivial eigenfunctions at $\ell =4$.}
\label{fig:eigenvalue-paths}
\end{center}
\end{figure}

\subsubsection{Hyperbolic manifolds with \texorpdfstring{$d \geq 3$}{TEXT}%
}
We can find similar upper bounds for closed hyperbolic manifolds of dimension $d>2$. We show these bounds in Fig.~\ref{fig:eigenvalue-bounds-2} for $d=3,4,5,6$, together with the $d=2$ bound just discussed. For small $\lambda_1$, the lines exhibit a change of slope but have no visible kinks, which in the conformal bootstrap are often associated with interesting theories \cite{Rychkov:2009ij, ElShowk:2012ht}.\footnote{There are kinks in the first derivatives of the boundary curves for $d\leq 4$, but we do not know if this is significant.} 
The explicit eigenvalues for hyperbolic 3-manifolds  computed in Refs.~\cite{Inoue:1998nz, Cornish1999} are easily compatible with the $d=3$ bound for larger $\lambda_1$.

More generally, we can look for an upper on $\lambda_{i+1}$ for a given value of $\lambda_i$ for any $i\geq 1$. For candidate eigenvalues $(\lambda_i^*, \lambda_{i+1}^*)$, we look for an $\vec{\alpha} \in \mathbb{R}^n$ such that
\begin{align}
\vec{\alpha} \cdot \vec{F}_{0}(\lambda_i^*, 0) & =1, \\
\vec{\alpha} \cdot \vec{F}_{0}(\lambda_i^*, x) & \geq 0, \quad \forall x \in (0, \lambda_i^*] \cup [\lambda_{i+1}^*, \infty),  \\
\vec{\alpha} \cdot \vec{F}_{s}(\lambda_i^*, x) & \geq 0, \quad \forall x \geq s+d-2, \quad s=2, 4, \dots, \Lambda/2,
\end{align} 
which can again be formulated as a semidefinite program.
It turns out that the resulting bounds are identical to those in Fig.~\ref{fig:eigenvalue-bounds-2} for the values of $d$ and $\lambda_1$ shown, i.e, we can replace $\lambda_1 \rightarrow \lambda_i$ and $\lambda_2 \rightarrow \lambda_{i+1}$ in Fig.~\ref{fig:eigenvalue-bounds-2} for any integer $i\geq 1$. Since the upper bounds grow monotonically, they are most constraining for low-lying eigenvalues.

\begin{figure}[h!t]
\begin{center}
\epsfig{file=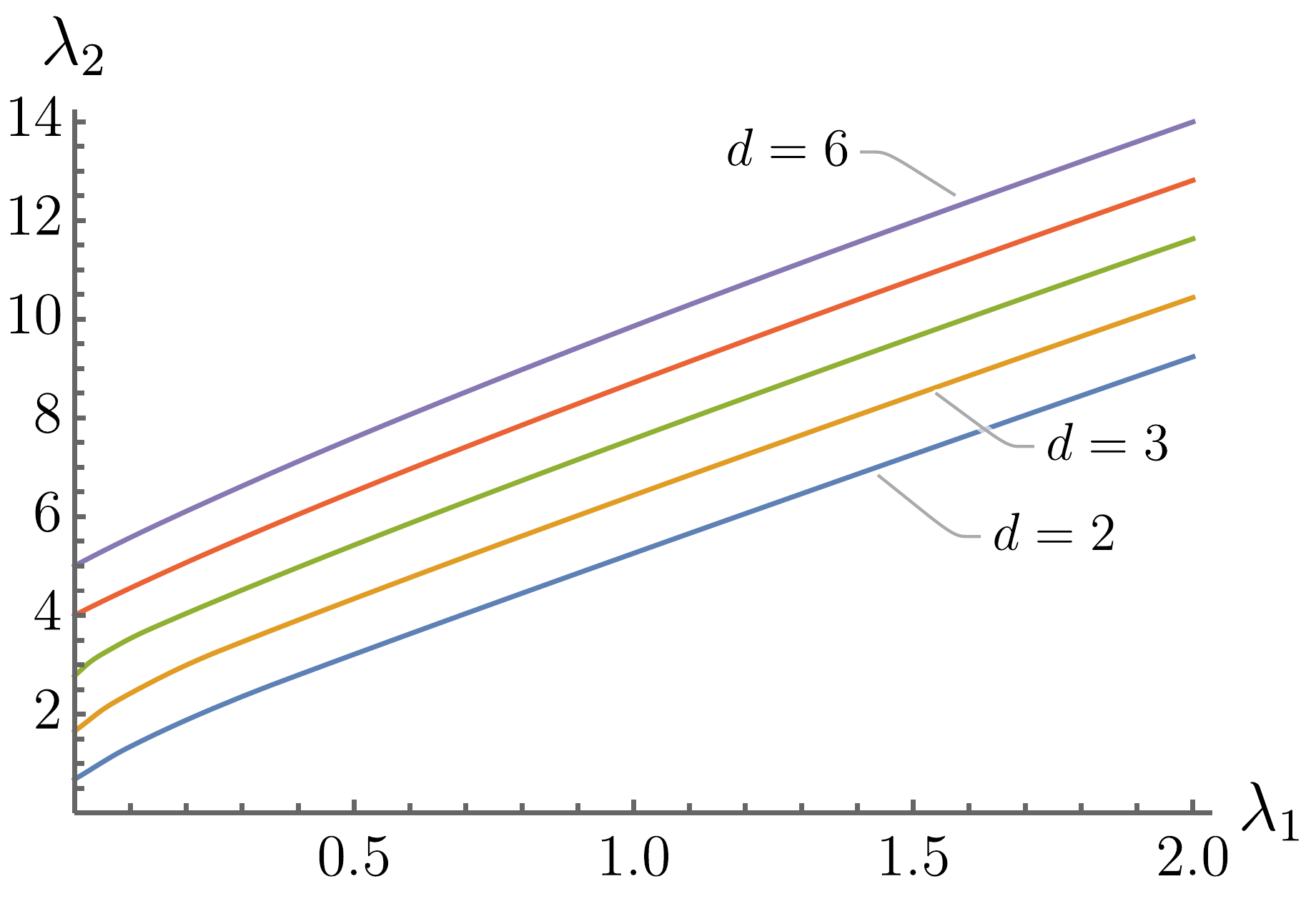, width=9cm}
\caption{Upper bounds on  the second nonzero eigenvalue of the Laplace--Beltrami operator, in terms of the first nonzero eigenvalue, on closed hyperbolic manifolds of dimension $d\in\{2,3,4,5,6\}$. The same bounds apply for any pair of consecutive nonzero eigenvalues.}
\label{fig:eigenvalue-bounds-2}
\end{center}
\end{figure}

We can use the same approach to find upper bounds on the smallest eigenvalue of the rough Laplacian on symmetric, transverse-traceless, rank-$s$ tensors for $d\geq 3$. We show these bounds for $s=2$ and $d=3, 4, 5 ,6$ in Fig.~\ref{fig:spin-2-bound}. These bounds have a nontrivial shape: they increase monotonically for $\lambda_1 \leq 3(d-1)^2/16$, then there is no bound for some range of $\lambda_1$, and beyond this the bounds are finite but much larger than the bounds at small $\lambda_1$. 

\begin{figure}[h!t]
\begin{center}
\epsfig{file=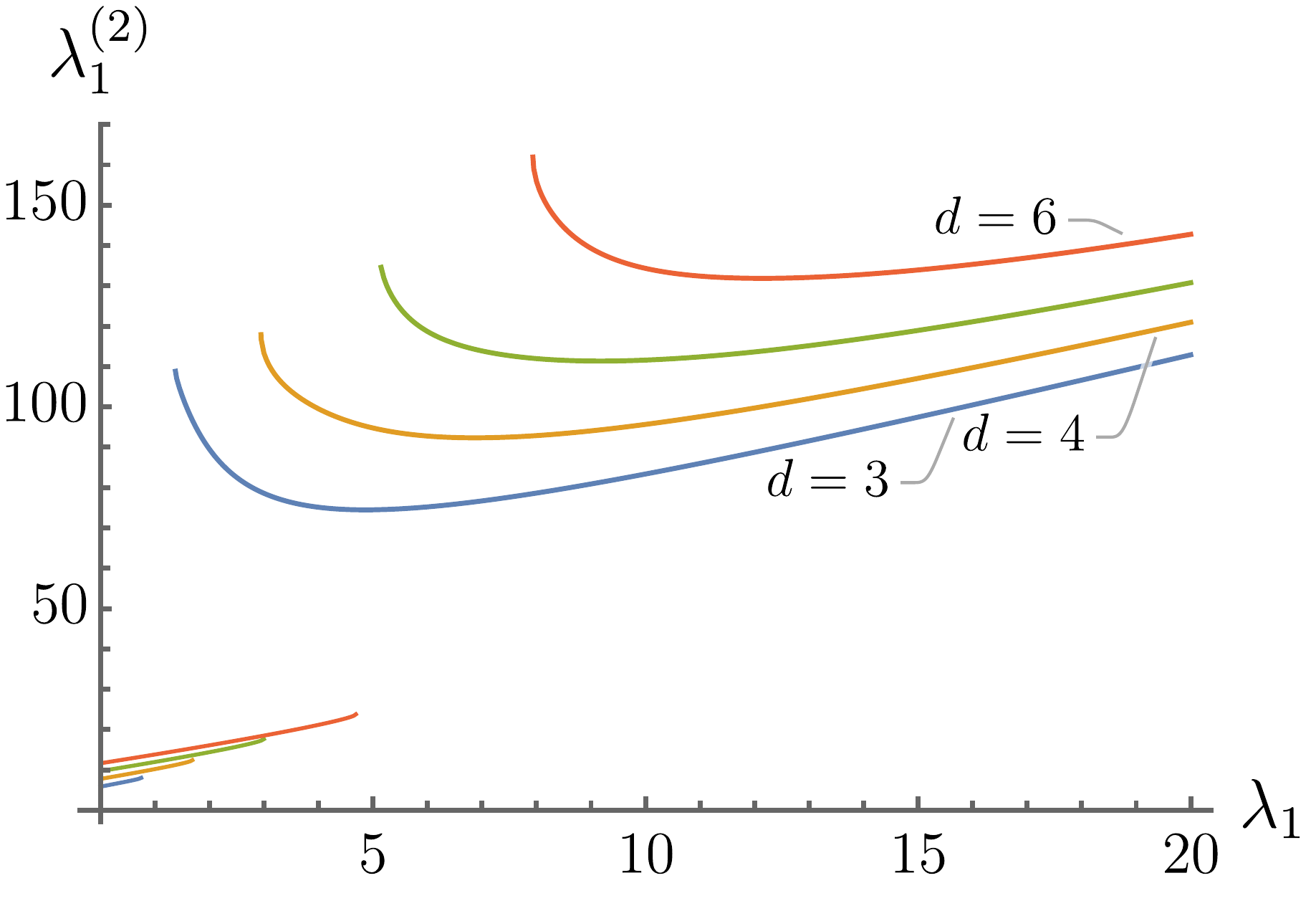, width=9cm}
\caption{Upper bounds on the smallest eigenvalue of the rough Laplacian on symmetric, transverse-traceless, rank-$2$ tensors for closed hyperbolic manifolds of dimension $d\in\{3,4,5,6\}$.}
\label{fig:spin-2-bound}
\end{center}
\end{figure}

\section{Bounds on triple overlap integrals}
\label{sec:overlap-bounds} 

We have so far discussed bounds on eigenvalues, but the consistency conditions can also be used to find bounds on integrals of products of eigenfunctions on hyperbolic manifolds. There is interest in analytic number theory in such overlap integrals and bounding their asymptotic growth due to their relation with automorphic $L$-functions in the case of certain non-compact manifolds \cite{Sarnak94, Petridis95, Bernstein1999, Bernstein2006}. 
The triple overlap integrals $c_{iij}$ decay faster than any polynomial in $\lambda_j$ due to the smoothness of the eigenfunctions, which guarantees the convergence of the sums in the consistency conditions. 
An explicit bound due to Sarnak \cite{Sarnak94} states that on a given hyperbolic manifold there are constants $A_i$ and $B$ such that
\be
|c_{iij}| \leq A_i (\lambda_j+1)^{B} e^{- \pi \sqrt{\lambda_j}/2}, \quad \forall j \in \mathbb{Z}_{>0}, 
\ee
where for $d=3$ we can take $B = 3/2$ \cite{Sarnak94}. In two dimensions, the best result was given by Bernstein and Reznikov \cite{Bernstein1999} and implies that on a given manifold there are constants $A_i$ such that\footnote{We have shifted the argument of the log compared to Ref.~\cite{Bernstein1999}.} 
\be \label{eq:Bernstein-Reznikov}
|c_{iij}| \leq A_i \left( \log (\lambda_j+1) \right)^{3/2} e^{- \pi \sqrt{\lambda_j}/2}, \quad \forall j \in \mathbb{Z}_{>0}.
\ee
These bounds apply to compact  hyperbolic manifolds and, for appropriate eigenfunctions, to non-compact, finite volume hyperbolic manifolds.
In this section, rather than bounding the asymptotic decay of triple overlap integrals, we instead find numerical bounds on triple overlap integrals of the lightest eigenfunction either with itself, with holomorphic $s$-differentials, or with moderately light eigenfunctions. We also show that some of these bounds are nearly saturated by genus-2 surfaces.

\subsection{Integrals of products of eigenfunctions}
\label{ssec:cubic-bounds}
We follow the conformal bootstrap approach for bounding the operator product expansion (OPE) coefficients of light states \cite{Caracciolo:2009bx}. We wish to find an upper bound on $\sqrt{V} |c_{111}|$, the rescaling-invariant combination of the volume and the triple overlap integral of the lightest non-constant eigenfunction. The approach is again to try to find contradictions by considering specific linear combinations of the consistency conditions.
Taking $\lambda_1 = \lambda_1^*$, we look for an $\vec{\alpha}$ satisfying the conditions
\begin{align}
\vec{\alpha} \cdot \vec{F}_{0}(\lambda_1^*, \lambda^*_1) & =1, \\
\vec{\alpha} \cdot \vec{F}_{0}(\lambda_1^*, x) & \geq 0, \quad \forall x \in  [\lambda_1^*, \infty), \\
\vec{\alpha} \cdot \vec{F}_{s}(\lambda_1^*, x) & \geq 0, \quad \forall x \geq s+d-2, \quad s=2, 4, \dots, \Lambda/2.
\end{align} 
It then follows from Eq.~\eqref{eq:sumRulesContracted}  that
\be 
V c^2_{111} \leq V \sum_{ \substack{\lambda_j = \lambda_1}} c_{11j}^2  \leq - \chi_{0, \Lambda} (\lambda_1) \vec{\alpha} \cdot \vec{F}_{0}(\lambda_1^*, 0),
\ee
where the sum is over a basis of the first nontrivial eigenspace.
To find the best upper bound, we maximise $\vec{\alpha} \cdot \vec{F}_{0}(\lambda_1^*, 0)$ subject to these constraints, which is a problem that can be solved using semidefinite programming \cite{Poland:2011ey}. We show the results for $d=2, \dots, 6$ in Fig.~\ref{fig:coupling-bound-1}. The bounds all diverge below some value of $\lambda_1$. None of the genus-2 surfaces that we checked were close to saturating the $d=2$ bound.

\begin{figure}[h!t]
\begin{center}
\epsfig{file=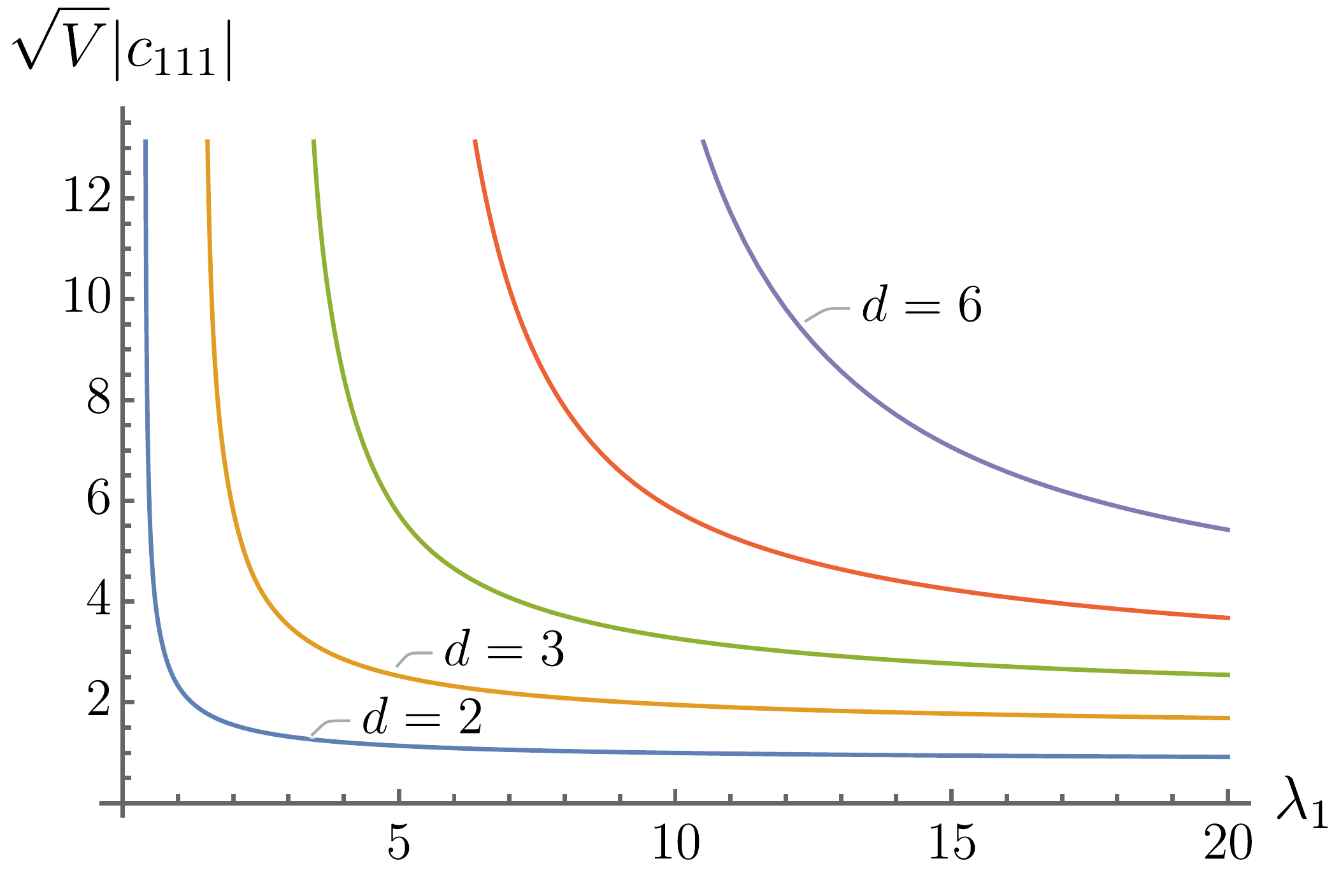, width=9cm}
\caption{Upper bounds on the triple overlap integrals of the lightest nontrivial scalar eigenfunctions on closed hyperbolic manifolds of dimension $d\in\{2,3,4,5, 6\}$.}
\label{fig:coupling-bound-1}
\end{center}
\end{figure}

\subsection{Eigenfunctions and holomorphic \texorpdfstring{$s$}{TEXT}%
-differentials}
We can also find bounds for overlap integrals involving transverse-traceless tensors. Here we only consider $d=2$. For a closed hyperbolic surface of genus $g$, it is notationally convenient to think of $c_{ijk}^{(s)}$ for $s>1$ as the $k$\textsuperscript{th} component of a finite-dimensional vector $\vec{c}_{ij}^{\, (s)} $ with norm given by
\be
 \lVert \vec{c}_{ij}^{\, (s)} \rVert= \sqrt{\sum_{k=1}^{2(2s-1)(g-1)} \left(c^{(s)}_{ijk}\right)^2}.
\ee
Recall that $c_{ijk}^{(s)}$ can be written in terms of the real part of the triple overlap integral of eigenfunctions and holomorphic $s$-differentials, as in Eq.~\eqref{eq:overlap2D}.
In this notation, the $d=2$ consistency conditions are written as
\be \label{eq:2DSumRules}
V^{-1} \vec{\alpha} \cdot \vec{F}_{0}(\lambda_i, 0) +\sum_{j=1}^{\infty} \frac{ \vec{\alpha} \cdot \vec{F}_{0}(\lambda_i,\lambda_j)}{\chi_{0, \Lambda} (\lambda_j)} c_{i i j}^2 + \sum_{\substack{s=2 \\ s \, \text{even}}}^{\Lambda/2} \frac{ \vec{\alpha} \cdot \vec{F}_{s}(\lambda_i,s)}{\chi_{s, \Lambda} ( s)}  \lVert \vec{c}_{ii}^{\, (s)} \rVert^2 =0.
\ee

Now suppose we want to look for bounds on the quantity $\sqrt{V}  \lVert \vec{c}_{11}^{\, (s')} \rVert$ for some $s'>1$. Taking $\lambda_1 = \lambda_1^*$, we look for an $\vec{\alpha}$ such that
\begin{align}
\vec{\alpha} \cdot \vec{F}_{s}(\lambda_1^*, s') & = \delta, \\
\vec{\alpha} \cdot \vec{F}_{0}(\lambda_1^*, x) & \geq 0, \quad \forall x \in  [\lambda_1^*, \infty), \\
\vec{\alpha} \cdot \vec{F}_{s}(\lambda_1^*, s) & \geq 0, \quad s \in \{2, 4, \dots, \Lambda/2 \} \setminus \{s'\}, 
\end{align} 
where $\delta = 1$ for upper bounds and $\delta=-1$ for lower bounds. For $\delta =1$ the consistency conditions in Eq.~\eqref{eq:2DSumRules} give the inequality
\be  
V \lVert \vec{c}_{11}^{\, (s')} \rVert^2  \leq - \chi_{s', \Lambda} ( s')\vec{\alpha} \cdot \vec{F}_{0}(\lambda_1^*, 0)
\ee
and for $\delta =-1$ they give
\be
V \lVert \vec{c}_{11}^{\, (s')} \rVert^2 \geq \chi_{s', \Lambda} ( s') \vec{\alpha} \cdot \vec{F}_{0}(\lambda_1^*, 0) .
\ee
In each case, we find the strongest bound by searching for an $\vec{\alpha}$ subject to the above constraints that maximises $\vec{\alpha} \cdot \vec{F}_{0}(\lambda_1^*, 0)$.

We show in Figs.~\ref{fig:spin-2-coupling-bound-2d} and \ref{fig:spin-4-coupling-bound-2d} the bounds we get on $\sqrt{V} \lVert \vec{c}_{11}^{\, (s)} \rVert$ for $s=2$ and $4$, together with the values for some explicit genus-2 surfaces along the path $\rho$ through Teichm\"uller space, defined in Eq.~\eqref{eq:moduli-path}. We explain  below  how to compute the overlap integrals for these surfaces. The upper bounds have similar shapes and diverge for $\lambda_1 \approx 0.399$, a number that we expect to decrease as $\Lambda$ is increased.  For $s=2$, there is a nontrivial lower bound, which exists for $\lambda_1 \leq 3/16$. The appearance of $3/16$ here can be traced back to the zero of $\vec{F}_0$ in Eq.~\eqref{eq:3/16zero}. Some of the explicit examples get quite close to saturating the bounds, e.g., the genus-2 surface $\rho(5)$ has $\sqrt{V} \lVert \vec{c}_{11}^{\, (2)} \rVert \approx  0.865$ and $\lambda_1 \approx 0.948$, while the bound at this value of $\lambda_1$ is $\sqrt{V} \lVert \vec{c}_{11}^{\, (2)} \rVert \leq  0.871$.  

\begin{figure}[ht!]
\begin{center}
\epsfig{file=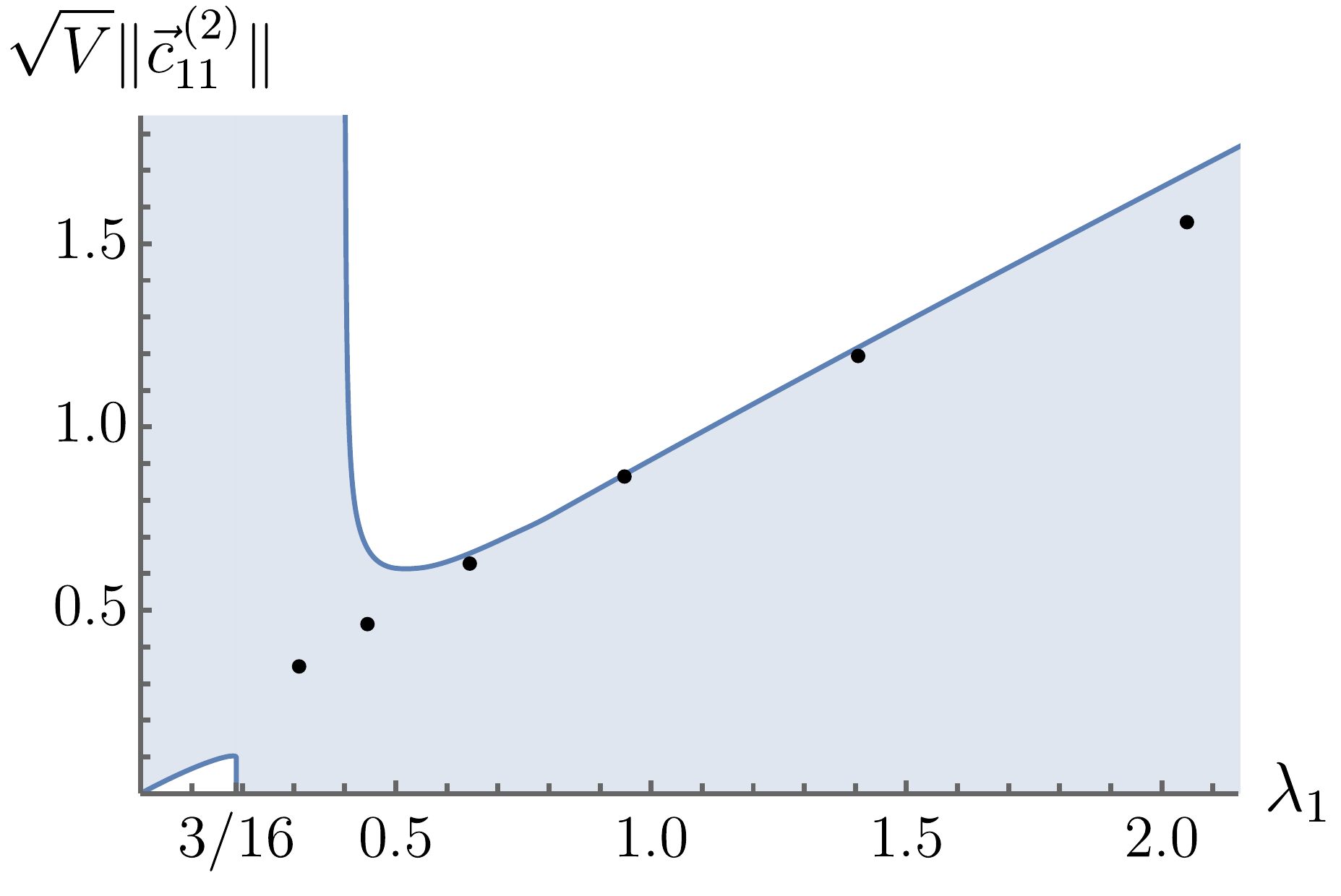, width=9cm}
\caption{Upper and lower bounds on the real parts of the overlap integrals of the lightest nontrivial eigenfunctions with holomorphic quadratic differentials on closed hyperbolic surfaces, as measured by $\sqrt{V} \lVert \vec{c}_{11}^{\, (2)} \rVert$. The blue region is allowed by the bounds. The filled circles correspond to the genus-2 surfaces $\rho(\ell)$, defined in Eq.~\eqref{eq:moduli-path}, for $\ell =3, \dots,8$, going from right to left.}
\label{fig:spin-2-coupling-bound-2d}
\end{center}
\end{figure}

\begin{figure}[h!t]
\begin{center}
\epsfig{file=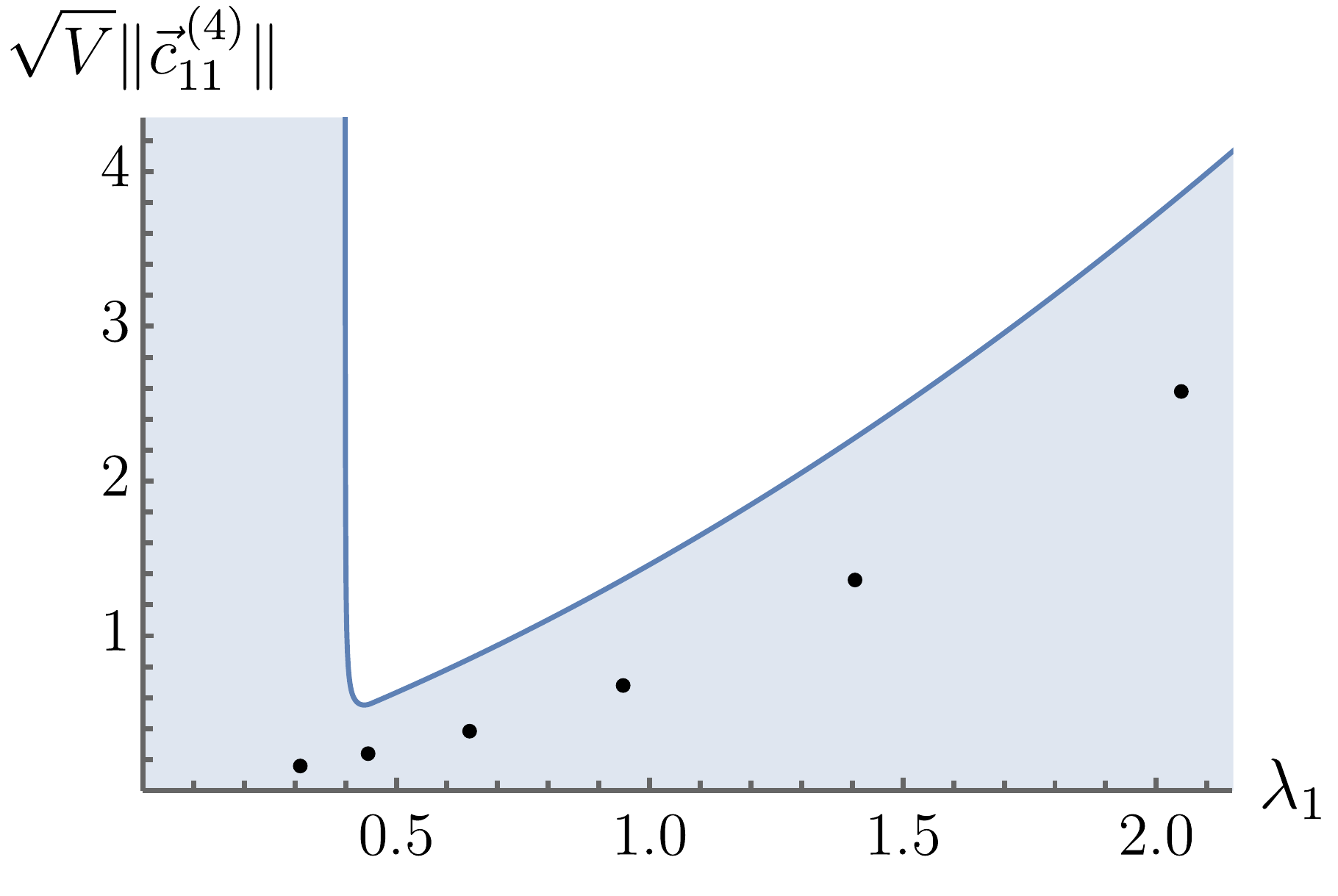, width=9cm}
\caption{An upper bound on $\sqrt{V} \lVert \vec{c}_{11}^{\, (4)} \rVert$ for closed hyperbolic surfaces. The filled circles correspond to the genus-2 surfaces $\rho(\ell)$, defined in Eq.~\eqref{eq:moduli-path},  for $\ell =3, \dots ,8$, going from right to left.}
\label{fig:spin-4-coupling-bound-2d}
\end{center}
\end{figure}

\subsubsection{Genus-2 examples}
We now explain how to find the genus-2 overlap integrals $\lVert \vec{c}_{11}^{\, (s)} \rVert$ included in Figs.~\ref{fig:spin-2-coupling-bound-2d} and \ref{fig:spin-4-coupling-bound-2d}. Rather than directly compute these integrals, we instead compute the scalar overlap integrals $c_{11j}$, plug these into the consistency conditions \eqref{eq:2DSumRules}, and then solve for $\lVert \vec{c}_{11}^{\, (s)} \rVert$.\footnote{In two dimensions, there are always two consistency conditions with contributions from the same maximum value of $s$, so verifying that they give the same value for $\lVert \vec{c}_{11}^{\, (s)} \rVert$ is a nontrivial check of the consistency conditions. 
} To evaluate the integrals $c_{11j}$ for a given genus-2 surface, we use the fundamental domain and Fourier coefficients computed by \texttt{Hypermodes} \cite{hypermodes, Strohmaier_2012}.  From the Fourier coefficients, we can construct approximations to the eigenfunctions as linear combinations of eigenfunctions on a hyperbolic cylinder \cite{Strohmaier_2012}, which can then be numerically integrated over the fundamental domain in the upper half-plane. 

It is important to estimate the accuracy of the resulting determinations of $\lVert \vec{c}_{11}^{\, (s)} \rVert$. The most significant contributions to the errors are the following: 
\begin{itemize}
\item We can only evaluate finitely many of the overlap integrals $c_{11j}$. In practice, we restrict to eigenfunctions with $\lambda_j \leq \lambda_{\text{max}}$, where $\lambda_{\text{max}}$ is some cutoff. To estimate the truncation error from taking $ \lambda_{\rm max}<\infty$, we use the bound \eqref{eq:Bernstein-Reznikov} of Ref.~\cite{Bernstein1999} to make the following nonrigorous estimate:
\be
\left| \sum_{\lambda_j>\lambda_{\rm{max}}} \frac{\vec{F}_{0}(\lambda_1,\lambda_j)}{\chi_{0, \Lambda} (\lambda_j)} c_{11j}^2  \right| \lesssim \frac{A^2}{2} \int_{\lambda_{\rm max}}^{\infty} d \lambda \left| \frac{\vec{F}_{0}(\lambda_1,\lambda)}{\chi_{0, \Lambda} (\lambda)} \right|  \left( \log (\lambda+1) \right)^{3} e^{- \pi \sqrt{\lambda}},
\ee
where the eigenvalue density comes from Weyl's law, with an extra factor of $1/2$ since only $\mathbb{Z}_2$-even eigenfunctions can have nonzero $c_{11j}$, and $A$ is estimated from the explicitly computed integrals. We found that $A \in [2,6]$ gives a good upper bound for our examples (see Fig.~\ref{fig:heavy-spin-0-exchange}).
\item The accuracy of the eigenfunction approximation. This is determined by a parameter called $N$ in \texttt{Hypermodes}, where the genus-2 eigenfunctions are approximated by linear combinations of $2(2N+1)$ eigenfunctions on a hyperbolic cylinder. It is also important that enough points are sampled by \texttt{Hypermodes} when it normalises the eigenfunctions. We checked explicitly that the reconstructed eigenfunctions were normalised to high accuracy.
\item The accuracy of the integrals $c_{11j}$ for $\lambda_j \leq \lambda_{\text{max}}$. We numerically integrated using the function \texttt{NIntegrate} in \texttt{Mathematica} with the option $\text{AccuracyGoal} \rightarrow a$. This means that the target error in the integrals is $10^{-a}$. We took this as the actual error in the nonzero $c_{11j}$, although this is sometimes an underestimate.
\end{itemize}
To find sufficiently accurate estimates of $\lVert \vec{c}_{11}^{\, (2)} \rVert$ and $\lVert \vec{c}_{11}^{\, (4)} \rVert$, we take $\lambda_{\text{max}} = 80 $, $a \in \{ 5, 6\}$, and $N \in \{ 30, 40, 50\}$, depending on the surface.  The resulting error estimates for the examples in Figs.~\ref{fig:spin-2-coupling-bound-2d} and \ref{fig:spin-4-coupling-bound-2d} are all smaller than the size of the plot markers and can be found in the ancillary file, together with the values of $c_{11j}$. We also show in Fig.~\ref{fig:heavy-spin-0-exchange} a log plot with the largest overlap integrals $|c_{11j}|$ of the genus-2 surface $\rho(5)$, together with a bootstrap bound on $\sqrt{V}|c_{11j}|$ for $\lambda_1 = 0.9477$ and the Bernstein--Reznikov bound  \eqref{eq:Bernstein-Reznikov} with $A_1=3$. The surface $\rho(5)$ has $\lambda_4\approx4.96$ and $\sqrt{V}|c_{114}| \approx0.582$, while the bootstrap upper bound at this eigenvalue is $0.602$.
\begin{figure}[h!t]
\begin{center}
\epsfig{file=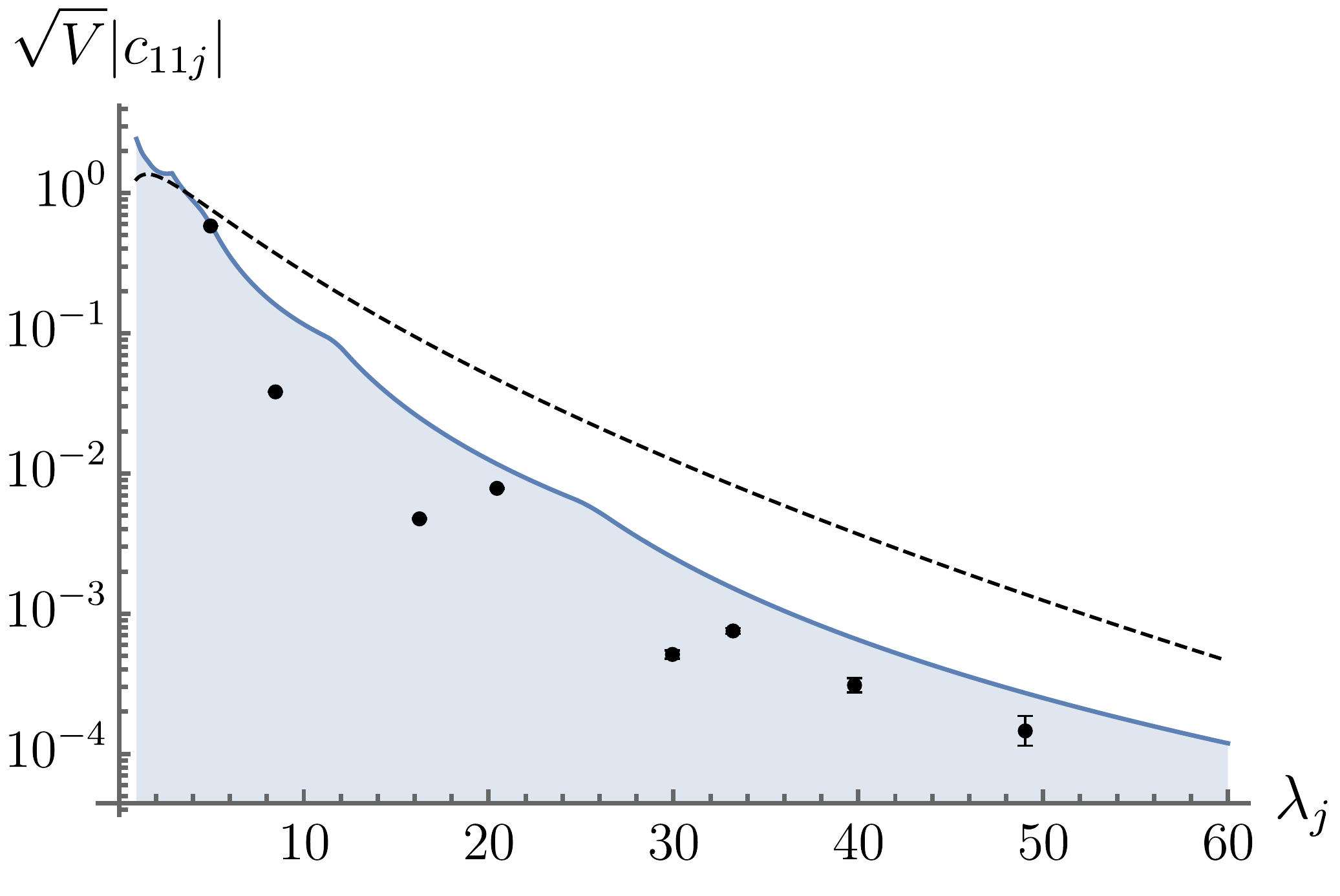, width=8.425cm}
\caption{An upper bound on $\sqrt{V} |c_{11j}|$ for closed hyperbolic surfaces with $\lambda_1 = 0.9477$. The filled circles represent the eigenfunctions $\phi_j$ of the genus-2 surface $\rho(5)$, defined in Eq.~\eqref{eq:moduli-path}, with $j\in\{4, 7, 16, 18, 29, 33, 39, 48\}$. The dashed line is the Bernstein--Reznikov bound \eqref{eq:Bernstein-Reznikov} with $A_1=3$, using $V=4 \pi$.}
\label{fig:heavy-spin-0-exchange}
\end{center}
\end{figure}

\section{Discussion}
\label{sec:discussion}
The application of bootstrap methods to hyperbolic manifolds could be extended in many directions. For many conformal bootstrap ideas, there is an analogue for closed hyperbolic manifolds, such as including multiple external states \cite{Kos:2014bka}, using external states with spin \cite{Iliesiu:2015qra, Dymarsky:2017xzb}, and assuming discrete global symmetries \cite{ElShowk:2012ht, Rong:2017cow, Baggio:2017mas}. The rigidity of hyperbolic 3-manifolds would make these good targets to try to isolate using a multi-scalar bootstrap. Optimistically, one could hope that bootstrap methods might even be useful for classifying hyperbolic manifolds.

The single scalar bootstrap of this paper could also be pushed much further. The current limiting factor is the difficulty of finding consistency conditions from integrands with many derivatives. In the conformal bootstrap, the bootstrap equation is a functional equation depending on the conformal cross-ratios, and discrete consistency conditions are usually extracted by doing a Taylor series expansion around a point. Here we instead constructed the consistency conditions separately at each derivative order. The effort required with this approach grows rapidly with the number of derivatives, so having a functional constraint would be very useful. Perhaps the mysterious zero in Eq.~\eqref{eq:3/16zero} is a hint that there is a better way to derive the consistency conditions.
However, we expect that further progress could be made even with the current approach, especially in two dimensions by further exploiting the complex structure.

It would be helpful to clarify the relationship between this geometric bootstrap and the conformal bootstrap, both for conceptual and technical reasons. 
The spectrum of a CFT is generally much more complicated than that of a manifold, but perhaps some limit of the conformal bootstrap gives geometric bounds, for example, through 2D $\sigma$-model CFTs (see, e.g., Refs.~\cite{Keller:2012mr, Lin:2015wcg}).\footnote{We thank Xi Yin for comments on this.} 
For hyperbolic manifolds, the equivalence of the symmetries of  $\mathbb{H}^{d+1}$ and Euclidean CFT$_d$ is also suggestive.
For Ricci-flat manifolds, there is an interpretation of the consistency conditions as sum rules for graviton amplitudes \cite{Bonifacio:2019ioc}. Perhaps there is a higher-spin analogue of this interpretation for hyperbolic manifolds.  

Lastly, it would be interesting if bootstrap bounds could be derived for non-compact finite volume hyperbolic manifolds, such as modular curves and knot complements. This would be somewhat analogous to the conformal bootstrap for non-compact CFTs \cite{guillarmou2020}. A non-compact bootstrap would require more complicated spectral decompositions, such as including Eisenstein series for the continuous part of the spectrum. In certain cases, bounds on overlap integrals would then translate to bounds on automorphic $L$-functions \cite{Sarnak94, Petridis95, Bernstein1999, Bernstein2006}. 

\paragraph{Acknowledgements:} I would like to thank Austin Joyce and Xi Yin for helpful comments and Kurt Hinterbichler for many helpful discussions and previous collaborations related to this topic. During this work I made use of the \texttt{Mathematica} package \texttt{xAct} \cite{xAct}. I am supported by the research program VIDI with Project No. 680-47-535, which is partly financed by the Netherlands Organisation for Scientific Research (NWO). This work has also been partially supported by STFC HEP consolidated grants ST/P000681/1 and ST/T000694/1.

\appendix
\section{Transverse-traceless tensor decomposition}
\label{app:ttDecomposition}
In this appendix, we prove the existence of a transverse-traceless decomposition for symmetric tensors of arbitrary rank on a closed $d$-dimensional Riemannian manifold $(\mathcal{M}, \hat{g})$. We could not find an explicit statement of this decomposition elsewhere, but it essentially follows from arguments in Refs.~\cite{Pestov-Sharafutdinov, Dairbekov2011, York1973, York1974}, which we follow closely.

Fix $s\geq 2$. Let $S^{s}(\mathcal{M})$ denote the space of smooth, symmetric, rank-$s$ tensor fields on $(\mathcal{M}, \hat{g})$ and let $S_0^{s}(\mathcal{M}) \subset S^{s}(\mathcal{M})$ denote the subspace of smooth, symmetric, traceless tensor fields.  Given a symmetric tensor $T^{(s)} \in S^{s}(\mathcal{M})$, we want to show that it can be uniquely decomposed into the sum of a symmetric, transverse-traceless tensor, the symmetrised traceless derivative of a traceless tensor, and a trace term. In components, this decomposition is  given by
\be \label{eq:ttDecomposeAppendix}
T^{(s)}_{m_1 \dots m_s} = T^{(s)TT}_{m_1 \dots m_s} + \nabla_{(m_1} W^{(s-1)}_{m_2 \dots m_s)}- \frac{s-1}{2s+d-4}\hat{g}_{(m_1 m_2} \nabla^{n}  W^{(s-1)}_{m_3 \dots m_s )n}+\hat{g}_{(m_1 m_2} \bar{T}^{(s-2)}_{m_3 \dots m_s)},
\ee
where $\nabla$ is the covariant derivative, $T^{(s)TT} \in S_0^{s}(\mathcal{M})$ is transverse, $ W^{(s-1)}\in S_0^{s-1}(\mathcal{M})$, and $\bar{T}^{(s-2)}\in S^{s-2}(\mathcal{M})$. We symmetrise with weight one. The different parts of this decomposition are orthogonal with respect to the $L^2$ inner product on symmetric tensors.  
For $s=2$, this is the familiar transverse-traceless decomposition of a rank-$2$ tensor, which is useful for studying the initial value formulation of general relativity \cite{York1973, York1974} and, in two dimensions, for calculating perturbative string amplitudes \cite{dHokerPhong88}. 

The tensor $\bar{T}^{(s-2)}$ can be expressed in terms of traces of $T^{(s)}$, while the symmetric, transverse-traceless tensor $T^{(s)TT}$ is defined by  Eq.~\eqref{eq:ttDecomposeAppendix}. We therefore only need to show that $W^{(s-1)}$ always exists and is unique up to terms that drop out of the decomposition. Taking the divergence of Eq.~\eqref{eq:ttDecomposeAppendix}, we define $W^{(s-1)}$ as the solution to the following equation:
\be
\label{eq:2ndOrder}
(P^\dagger_{s-1} P_{s-1}W^{(s-1)})_{m_2 \dots m_s}= -\nabla^{m_1} \left(T^{(s)}_{m_1 \dots m_s} - \hat{g}_{(m_1 m_2} \bar{T}^{(s-2)}_{m_3 \dots m_s)}\right),
\ee
where the operator $P_{s-1}$ is the symmetrised traceless derivative defined by
\be \label{eq:PsDef}
P_{s-1}: \,  S_0^{s-1}(\mathcal{M}) \rightarrow S_0^{s}(\mathcal{M}), \quad W^{(s-1)}_{m_1 \dots m_{s-1}} \mapsto  \nabla_{(m_1} W^{(s-1)}_{m_2 \dots m_s)}- \frac{s-1}{2s+d-4}\hat{g}_{(m_1 m_2} \nabla^n W^{(s-1)}_{m_3 \dots m_s) n},
\ee
and its formal adjoint is minus the divergence,
\be
P^{\dagger}_{s-1}: \,  S_0^{s}(\mathcal{M}) \rightarrow S_0^{s-1}(\mathcal{M}), \quad X^{(s)}_{m_1 \dots m_{s}} \mapsto  -\nabla^{m_s} X^{(s)}_{m_1 \dots m_s}.
\ee
In two dimensions, $P_{s-1}$ is the same as the operator $P_{s-1}=\nabla^z_{s-1}\oplus \nabla_z^{-s+1}$ of Ref.~\cite{Alvarez:1982zi}.  

We want to show that the operator $P^\dagger_{s-1} P_{s-1}$ is elliptic, since then we can invert it  in Eq.~\eqref{eq:2ndOrder} to solve for $W^{(s-1)}$.
Define the operator $j^{(s-1)}_{\xi} $ that contracts a symmetric, traceless, rank-$s$ tensor at a point with the constant vector $\xi$,
\be
j^{(s-1)}_{\xi} : \,  V_{m_1 \dots m_s} \mapsto  \xi^{m_s} V_{m_1 \dots m_s}, 
\ee
and its adjoint $i_{\xi}^{(s-1)}$,
\be
i_{\xi}^{(s-1)}: \, U_{m_1 \dots m_{s-1}} \mapsto \xi_{(m_1} U_{m_2 \dots m_s)} - \frac{s-1}{2s+d-4}\hat{g}_{(m_1 m_2} \xi^n U_{m_3 \dots m_s)n}.
\ee
The principal symbol of $P_{s-1}$ is $ \sigma(P_{s-1} , \xi)  = i \,i_{\xi}^{(s-1)}$, so the principal symbol of $P^\dagger_{s-1} P_{s-1}$  is
\begin{align}
\sigma(P^\dagger_{s-1} P_{s-1} , \xi) & =\sigma(P^\dagger_{s-1}, \xi) \sigma(P_{s-1} , \xi)=\sigma(P_{s-1}, \xi)^\dagger \sigma(P_{s-1} , \xi)  \\
& = j_{\xi}^{(s-1)} i^{(s-1)}_{\xi} = \frac{|\xi|^2}{s} {\rm Id}^{(s-1)} + \frac{(s-1)(2s+d-6)}{s(2s+d-4)}i_{\xi}^{(s-2)} j_{\xi}^{(s-2)},
\end{align}
where $|\xi|^2 = \xi_m \xi^m$ and ${\rm Id}^{(s-1)}$ is the identity.
Since $i_{\xi}^{(s-2)} j_{\xi}^{(s-2)}$ is non-negative and $|\xi|^2\, {\rm Id}^{(s-1)} $ is positive for nonzero $\xi$, the principal symbol $\sigma(P^\dagger_{s-1} P_{s-1} , \xi) $ defines an isomorphism for nonzero $\xi$. The operator $P^\dagger_{s-1} P_{s-1}$ is therefore strongly elliptic everywhere \cite{Dairbekov2011}.
The kernel of this operator consists of traceless conformal Killing tensors.  The right-hand side of Eq.~\eqref{eq:2ndOrder} is orthogonal to traceless conformal Killing tensors, so using the standard existence, uniqueness, and regularity theorems for the solutions of linear elliptic equations on closed manifolds---as summarised, for example, in the appendix of Ref.~\cite{Besse}---we can invert $P^\dagger_{s-1} P_{s-1}$ in Eq.~\eqref{eq:2ndOrder} to obtain a solution $W^{(s-1)}$ that is smooth and unique up to traceless conformal Killing tensors. Any traceless conformal Killing tensor contributions to $W^{(s-1)}$ drop out of Eq.~\eqref{eq:ttDecomposeAppendix}, so this establishes that this decomposition exists and is unique, and it is straightforward to check that it is orthogonal. Note that there exist no nonzero traceless conformal Killing tensors on closed hyperbolic manifolds \cite{Dairbekov2011}. 

\renewcommand{\em}{}
\bibliographystyle{utphys}
\addcontentsline{toc}{section}{References}
\bibliography{hyperbolic-refs}

\providecommand{\href}[2]{#2}\begingroup\raggedright\begin{thebibliography}{10}

\bibitem{Rattazzi:2008pe}
R.~Rattazzi, V.~S. Rychkov, E.~Tonni, and A.~Vichi, ``{Bounding scalar operator
  dimensions in 4D CFT},''
  \href{http://dx.doi.org/10.1088/1126-6708/2008/12/031}{{\em JHEP} {\bf 12}
  (2008)  031},
\href{http://arxiv.org/abs/0807.0004}{{\tt arXiv:0807.0004 [hep-th]}}.
%%CITATION = ARXIV:0807.0004;%%.

\bibitem{Rychkov:2009ij}
V.~S. Rychkov and A.~Vichi, ``{Universal Constraints on Conformal Operator
  Dimensions},'' \href{http://dx.doi.org/10.1103/PhysRevD.80.045006}{{\em Phys.
  Rev. D} {\bf 80} (2009)  045006}, \href{http://arxiv.org/abs/0905.2211}{{\tt
  arXiv:0905.2211 [hep-th]}}.

\bibitem{Caracciolo:2009bx}
F.~Caracciolo and V.~S. Rychkov, ``{Rigorous Limits on the Interaction Strength
  in Quantum Field Theory},''
  \href{http://dx.doi.org/10.1103/PhysRevD.81.085037}{{\em Phys. Rev.} {\bf
  D81} (2010)  085037},
\href{http://arxiv.org/abs/0912.2726}{{\tt arXiv:0912.2726 [hep-th]}}.
%%CITATION = ARXIV:0912.2726;%%.

\bibitem{Poland:2011ey}
D.~Poland, D.~Simmons-Duffin, and A.~Vichi, ``{Carving Out the Space of 4D
  CFTs},'' \href{http://dx.doi.org/10.1007/JHEP05(2012)110}{{\em JHEP} {\bf 05}
  (2012)  110},
\href{http://arxiv.org/abs/1109.5176}{{\tt arXiv:1109.5176 [hep-th]}}.
%%CITATION = ARXIV:1109.5176;%%.

\bibitem{Kos:2014bka}
F.~Kos, D.~Poland, and D.~Simmons-Duffin, ``{Bootstrapping Mixed Correlators in
  the 3D Ising Model},'' \href{http://dx.doi.org/10.1007/JHEP11(2014)109}{{\em
  JHEP} {\bf 11} (2014)  109},
\href{http://arxiv.org/abs/1406.4858}{{\tt arXiv:1406.4858 [hep-th]}}.
%%CITATION = ARXIV:1406.4858;%%.

\bibitem{Poland:2018epd}
D.~Poland, S.~Rychkov, and A.~Vichi, ``{The Conformal Bootstrap: Theory,
  Numerical Techniques, and Applications},''
  \href{http://dx.doi.org/10.1103/RevModPhys.91.015002}{{\em Rev. Mod. Phys.}
  {\bf 91} (2019)  015002}.

\bibitem{Bonifacio:2020xoc}
J.~Bonifacio and K.~Hinterbichler, ``{Bootstrap Bounds on Closed Einstein
  Manifolds},'' \href{http://dx.doi.org/10.1007/JHEP10(2020)069}{{\em JHEP}
  {\bf 10} (2020)  069}, \href{http://arxiv.org/abs/2007.10337}{{\tt
  arXiv:2007.10337 [hep-th]}}.

\bibitem{Bonifacio:2019ioc}
J.~Bonifacio and K.~Hinterbichler, ``{Unitarization from Geometry},''
  \href{http://dx.doi.org/10.1007/JHEP12(2019)165}{{\em JHEP} {\bf 12} (2019)
  165}, \href{http://arxiv.org/abs/1910.04767}{{\tt arXiv:1910.04767
  [hep-th]}}.

\bibitem{hypermodes}
A.~Strohmaier and V.~Uski, ``Hypermodes.''
  \url{http://www1.maths.leeds.ac.uk/~pmtast/hyperbolic-surfaces/hypermodes.html}.
\newblock Accessed: 22-01-2021.

\bibitem{Strohmaier_2012}
A.~{Strohmaier} and V.~{Uski}, ``{An Algorithm for the Computation of
  Eigenvalues, Spectral Zeta Functions and Zeta-Determinants on Hyperbolic
  Surfaces},'' \href{http://dx.doi.org/10.1007/s00220-012-1557-1}{{\em
  Communications in Mathematical Physics} {\bf 317} (2013) no.~3, 827--869},
  \href{http://arxiv.org/abs/1110.2150}{{\tt arXiv:1110.2150 [math.SP]}}.

\bibitem{Pestov-Sharafutdinov}
L.~Pestov and V.~Sharafutdinov, ``{Integral geometry of tensor fields on a
  manifold of negative curvature},''
  \href{http://dx.doi.org/https://doi.org/10.1007/BF00969652}{{\em Sib Math J}
  {\bf 29} (1988)  427--441}.

\bibitem{Dairbekov2011}
N.~S. Dairbekov and V.~Sharafutdinov, ``{On conformal Killing symmetric tensor
  fields on Riemannian manifolds},'' {\em Siberian Advances in Mathematics}
  {\bf 21} (2011)  1--41.

\bibitem{York1973}
J.~W. {York Jr.}, ``{Conformally invariant orthogonal decomposition of
  symmetric tensors on Riemannian manifolds and the initial-value problem of
  general relativity},'' \href{http://dx.doi.org/10.1063/1.1666338}{{\em
  Journal of Mathematical Physics} {\bf 14} (1973) no.~4, 456--464}.

\bibitem{York1974}
J.~W. York~Jr., ``Covariant decompositions of symmetric tensors in the theory
  of gravitation,'' {\em Annales de l'I.H.P. Physique th\'eorique} {\bf 21}
  (1974) no.~4, 319--332.

\bibitem{buser1992}
P.~Buser, {\em Geometry and Spectra of Compact Riemann Surfaces}.
\newblock Progress in mathematics. Springer, 1992.

\bibitem{Besse}
A.~L. Besse, {\em {Einstein Manifolds}}.
\newblock Springer-Verlag, Berlin, Heidelberg, New York,
1987.
\newblock
%%CITATION = INSPIRE-1235587;%%.

\bibitem{marden_2016}
A.~Marden, \href{http://dx.doi.org/10.1017/CBO9781316337776}{{\em Hyperbolic
  Manifolds: An Introduction in 2 and 3 Dimensions}}.
\newblock Cambridge University Press, 2016.

\bibitem{Otal-2009}
J.-P. Otal and E.~Rosas, ``Pour toute surface hyperbolique de genre $g$,
  $\lambda_{2g-2} > 1/4$,''
  \href{http://dx.doi.org/10.1215/00127094-2009-048}{{\em Duke Mathematical
  Journal} {\bf 150} (2009) no.~1, 101–115}.

\bibitem{Buser1977}
P.~Buser, ``{Riemannsche Fl{\"a}chen mit Eigenwerten in (0,1/4)},'' {\em
  Commentarii Mathematici Helvetici} {\bf 52} (1977)  25--34.

\bibitem{Huber1974}
H.~Huber, ``{Über den ersten Eigenwert des Laplace-Operators auf kompakten
  Riemannschen Flächen},'' {\em Commentarii mathematici Helvetici} {\bf 49}
  (1974)  251--259.

\bibitem{YangYau1980}
P.~Yang and S.~Yau, ``{Eigenvalues of the Laplacian of compact Riemann surfaces
  and minimal submanifolds},'' {\em Annali Della Scuola Normale Superiore Di
  Pisa-classe Di Scienze} {\bf 7} (1980)  55--63.

\bibitem{ElSoufi83}
A.~El~Soufi and S.~Ilias, ``{Le volume conforme et ses applications d'apr\`es
  Li et Yau},'' {\em S\'eminaire de th\'eorie spectrale et g\'eom\'etrie} {\bf
  2} (1983-1984)  .

\bibitem{Aurich1989}
R.~Aurich and F.~Steiner, ``{Periodic-orbit sum rules for the
  Hadamard-Gutzwiller model},''
  \href{http://dx.doi.org/https://doi.org/10.1016/0167-2789(89)90003-1}{{\em
  Physica D: Nonlinear Phenomena} {\bf 39} (1989) no.~2, 169--193}.

\bibitem{ros2021}
A.~Ros, ``{On the first eigenvalue of the Laplacian on compact surfaces of
  genus three},'' \href{http://arxiv.org/abs/2010.14857}{{\tt arXiv:2010.14857
  [math.DG]}}.

\bibitem{karpukhin2021}
M.~Karpukhin and D.~Vinokurov, ``{An improved Yang-Yau inequality for the first
  Laplace eigenvalue},'' \href{http://arxiv.org/abs/2106.00627}{{\tt
  arXiv:2106.00627 [math.DG]}}.

\bibitem{karpukhin2020}
M.~Karpukhin, N.~Nadirashvili, A.~V. Penskoi, and I.~Polterovich,
  ``{Conformally maximal metrics for Laplace eigenvalues on surfaces},''
  \href{http://arxiv.org/abs/2003.02871}{{\tt arXiv:2003.02871 [math.DG]}}.

\bibitem{lipnowski2021}
M.~Lipnowski and A.~Wright, ``Towards optimal spectral gaps in large genus,''
  \href{http://arxiv.org/abs/2103.07496}{{\tt arXiv:2103.07496 [math.GT]}}.

\bibitem{wu2021}
Y.~Wu and Y.~Xue, ``Random hyperbolic surfaces of large genus have first
  eigenvalues greater than $\frac{3}{16}-\epsilon$,''
  \href{http://arxiv.org/abs/2102.05581}{{\tt arXiv:2102.05581 [math.DG]}}.

\bibitem{Mirzakhani2013}
M.~Mirzakhani, ``{Growth of Weil-Petersson Volumes and Random Hyperbolic
  Surface of Large Genus},''
  \href{http://dx.doi.org/10.4310/jdg/1367438650}{{\em Journal of Differential
  Geometry} {\bf 94} (2013) no.~2, 267 -- 300}.

\bibitem{Wright2020}
A.~Wright, ``{A tour through Mirzakhani’s work on moduli spaces of Riemann
  surfaces},'' {\em Bulletin of the American Mathematical Society} {\bf 57}
  (2020)  359--408.

\bibitem{magee2020}
M.~Magee, F.~Naud, and D.~Puder, ``A random cover of a compact hyperbolic
  surface has relative spectral gap $\frac{3}{16}-\epsilon$,''
  \href{http://arxiv.org/abs/2003.10911}{{\tt arXiv:2003.10911 [math.SP]}}.

\bibitem{Thurston82}
W.~P. Thurston, ``{Three dimensional manifolds, Kleinian groups and hyperbolic
  geometry},'' {\em Bulletin of the American Mathematical Society} {\bf 6}
  (1982) no.~3, 357 -- 381.

\bibitem{Perelman2002}
G.~Perelman, ``{The entropy formula for the Ricci flow and its geometric
  applications},'' \href{http://arxiv.org/abs/math/0211159}{{\tt
  arXiv:math/0211159 [math.DG]}}.

\bibitem{Mostow68}
G.~D. Mostow, ``Quasi-conformal mappings in $n$-space and the rigidity of
  hyperbolic space forms,'' {\em Publications Math\'ematiques de l'IH\'ES} {\bf
  34} (1968)  53--104.

\bibitem{Agol2013}
I.~Agol, ``{The Virtual Haken Conjecture},'' {\em Documenta Mathematica} {\bf
  18} (2013)  1045--1087. With an appendix by Ian Agol, Daniel Groves, and
  Jason Manning.

\bibitem{Callahan-thesis}
P.~J. Callahan, {\em {Spectral geometry of hyperbolic 3-manifolds}}.
\newblock PhD thesis, University of Illinois at Urbana-Champaign, 1994.
\newblock \url{http://hdl.handle.net/2142/22992}.

\bibitem{Schoen1982}
R.~Schoen, ``A lower bound for the first eigenvalue of a negatively curved
  manifold,'' {\em Journal of Differential Geometry} {\bf 17} (1982)  233--238.

\bibitem{Inoue:1998nz}
K.~T. Inoue, ``{Computation of eigenmodes on a compact hyperbolic space},''
  \href{http://dx.doi.org/10.1088/0264-9381/16/10/304}{{\em Class. Quant.
  Grav.} {\bf 16} (1999)  3071--3094},
  \href{http://arxiv.org/abs/astro-ph/9810034}{{\tt arXiv:astro-ph/9810034}}.

\bibitem{Cornish1999}
N.~Cornish and D.~Spergel, ``On the eigenmodes of compact hyperbolic
  3-manifolds,'' \href{http://arxiv.org/abs/math/9906017}{{\tt
  arXiv:math/9906017 [math.DG]}}.

\bibitem{Inoue2001}
K.~T. Inoue, ``Numerical study of length spectra and low-lying eigenvalue
  spectra of compact hyperbolic 3-manifolds,'' {\em Classical and Quantum
  Gravity} {\bf 18} (2001)  629--652.

\bibitem{Gabai2007}
D.~Gabai, R.~Meyerhoff, and P.~Milley, ``Minimum volume cusped hyperbolic
  three-manifolds,'' {\em Journal of the American Mathematical Society} {\bf
  22} (2007)  1157--1215.

\bibitem{Hinterbichler_2014}
K.~Hinterbichler, J.~Levin, and C.~Zukowski, ``{Kaluza-Klein towers on general
  manifolds},'' \href{http://dx.doi.org/10.1103/physrevd.89.086007}{{\em
  Physical Review D} {\bf 89} (2014)  }.

\bibitem{Dyatlov2015}
S.~Dyatlov, F.~Faure, and C.~Guillarmou, ``Power spectrum of the geodesic flow
  on hyperbolic manifolds,''
  \href{http://dx.doi.org/10.2140/apde.2015.8.923}{{\em Analysis \& PDE} {\bf
  8} (2015)  923–1000}.

\bibitem{Simons68}
J.~Simons, ``{Minimal Varieties in Riemannian Manifolds},''
  \href{http://dx.doi.org/10.2307/1970556}{{\em Annals of Mathematics} {\bf 88}
  (1968) no.~1, 62--105}.

\bibitem{higuchi87}
A.~Higuchi, ``{Symmetric tensor spherical harmonics on the $N$‐sphere and
  their application to the de Sitter group $SO(N,1)$},''
  \href{http://dx.doi.org/10.1063/1.527513}{{\em Journal of Mathematical
  Physics} {\bf 28} (1987) no.~7, 1553--1566}.

\bibitem{Alvarez:1982zi}
O.~Alvarez, ``{Theory of Strings with Boundaries: Fluctuations, Topology, and
  Quantum Geometry},''
  \href{http://dx.doi.org/10.1016/0550-3213(83)90490-X}{{\em Nucl. Phys. B}
  {\bf 216} (1983)  125--184}.

\bibitem{dHokerPhong1986}
E.~{D'Hoker} and D.~H. {Phong}, ``{On determinants of Laplacians on Riemann
  surfaces},'' \href{http://dx.doi.org/10.1007/BF01211063}{{\em Communications
  in Mathematical Physics} {\bf 104} (1986) no.~4, 537--545}.

\bibitem{dHokerPhong88}
E.~D'Hoker and D.~H. Phong, ``The geometry of string perturbation theory,''
  \href{http://dx.doi.org/10.1103/RevModPhys.60.917}{{\em Rev. Mod. Phys.} {\bf
  60} (1988)  917--1065}.
  \url{https://link.aps.org/doi/10.1103/RevModPhys.60.917}.

\bibitem{Arefeva86}
I.~Aref'eva and I.~Volovich, ``{Hyperbolic manifolds as vacuum solutions in
  Kaluza-Klein theories},''
  \href{http://dx.doi.org/https://doi.org/10.1016/0550-3213(86)90530-4}{{\em
  Nuclear Physics B} {\bf 274} (1986) no.~3, 619--632}.

\bibitem{Kaloper:2000jb}
N.~Kaloper, J.~March-Russell, G.~D. Starkman, and M.~Trodden, ``{Compact
  hyperbolic extra dimensions: Branes, Kaluza-Klein modes and cosmology},''
  \href{http://dx.doi.org/10.1103/PhysRevLett.85.928}{{\em Phys. Rev. Lett.}
  {\bf 85} (2000)  928--931}, \href{http://arxiv.org/abs/hep-ph/0002001}{{\tt
  arXiv:hep-ph/0002001}}.

\bibitem{DeLuca:2021pej}
G.~B. De~Luca, E.~Silverstein, and G.~Torroba, ``{Hyperbolic compactification
  of M-theory and de Sitter quantum gravity},''
  \href{http://arxiv.org/abs/2104.13380}{{\tt arXiv:2104.13380 [hep-th]}}.

\bibitem{Sarnak94}
P.~Sarnak, ``{Integrals of products of eigenfunctions},''
  \href{http://dx.doi.org/10.1155/S1073792894000280}{{\em International
  Mathematics Research Notices} {\bf 1994} (1994) no.~6, 251--260}.

\bibitem{Petridis95}
Y.~N. Petridis, ``{ On squares of eigenfunctions for the hyperbolic plane and a
  new bound on certain L -series },''
  \href{http://dx.doi.org/10.1155/S1073792895000092}{{\em International
  Mathematics Research Notices} {\bf 1995} (1995) no.~3, 111--127}.

\bibitem{Bernstein1999}
J.~Bernstein and A.~Reznikov, ``Analytic continuation of representations and
  estimates of automorphic forms,'' {\em Annals of Mathematics} {\bf 150}
  (1999)  329--352.

\bibitem{Bernstein2006}
J.~Bernstein and A.~Reznikov, ``{Subconvexity bounds for triple L-functions and
  representation theory},'' {\em Annals of Mathematics} {\bf 172} (2006)
  1679--1718.

\bibitem{Csaki:2003dt}
C.~Csaki, C.~Grojean, H.~Murayama, L.~Pilo, and J.~Terning, ``{Gauge theories
  on an interval: Unitarity without a Higgs},''
  \href{http://dx.doi.org/10.1103/PhysRevD.69.055006}{{\em Phys. Rev. D} {\bf
  69} (2004)  055006}, \href{http://arxiv.org/abs/hep-ph/0305237}{{\tt
  arXiv:hep-ph/0305237}}.

\bibitem{Selberg3/16}
A.~Selberg, ``{On the estimation of Fourier coefficients of modular forms},''
  {\em Proceedings of Symposia in Pure Mathematics} {\bf 8} (1965)  .

\bibitem{Simmons-Duffin:2015qma}
D.~Simmons-Duffin, ``{A Semidefinite Program Solver for the Conformal
  Bootstrap},'' \href{http://dx.doi.org/10.1007/JHEP06(2015)174}{{\em JHEP}
  {\bf 06} (2015)  174}, \href{http://arxiv.org/abs/1502.02033}{{\tt
  arXiv:1502.02033 [hep-th]}}.

\bibitem{Landry:2019qug}
W.~Landry and D.~Simmons-Duffin, ``{Scaling the semidefinite program solver
  SDPB},'' \href{http://arxiv.org/abs/1909.09745}{{\tt arXiv:1909.09745
  [hep-th]}}.

\bibitem{Cook-thesis}
J.~Cook, {\em {Properties of eigenvalues on Riemann surfaces with large
  symmetry groups}}.
\newblock PhD thesis, Loughborough U., 2018.
\newblock \url{https://hdl.handle.net/2134/36294}.

\bibitem{strohmaier2016}
A.~Strohmaier, ``Computation of eigenvalues, spectral zeta functions and
  zeta-determinants on hyperbolic surfaces,''
  \href{http://arxiv.org/abs/1604.02722}{{\tt arXiv:1604.02722 [math.NA]}}.

\bibitem{ElShowk:2012ht}
S.~El-Showk, M.~F. Paulos, D.~Poland, S.~Rychkov, D.~Simmons-Duffin, and
  A.~Vichi, ``{Solving the 3D Ising Model with the Conformal Bootstrap},''
  \href{http://dx.doi.org/10.1103/PhysRevD.86.025022}{{\em Phys. Rev. D} {\bf
  86} (2012)  025022}, \href{http://arxiv.org/abs/1203.6064}{{\tt
  arXiv:1203.6064 [hep-th]}}.

\bibitem{Iliesiu:2015qra}
L.~Iliesiu, F.~Kos, D.~Poland, S.~S. Pufu, D.~Simmons-Duffin, and R.~Yacoby,
  ``{Bootstrapping 3D Fermions},''
  \href{http://dx.doi.org/10.1007/JHEP03(2016)120}{{\em JHEP} {\bf 03} (2016)
  120}, \href{http://arxiv.org/abs/1508.00012}{{\tt arXiv:1508.00012
  [hep-th]}}.

\bibitem{Dymarsky:2017xzb}
A.~Dymarsky, J.~Penedones, E.~Trevisani, and A.~Vichi, ``{Charting the space of
  3D CFTs with a continuous global symmetry},''
  \href{http://dx.doi.org/10.1007/JHEP05(2019)098}{{\em JHEP} {\bf 05} (2019)
  098}, \href{http://arxiv.org/abs/1705.04278}{{\tt arXiv:1705.04278
  [hep-th]}}.

\bibitem{Rong:2017cow}
J.~Rong and N.~Su, ``{Scalar CFTs and Their Large N Limits},''
  \href{http://dx.doi.org/10.1007/JHEP09(2018)103}{{\em JHEP} {\bf 09} (2018)
  103}, \href{http://arxiv.org/abs/1712.00985}{{\tt arXiv:1712.00985
  [hep-th]}}.

\bibitem{Baggio:2017mas}
M.~Baggio, N.~Bobev, S.~M. Chester, E.~Lauria, and S.~S. Pufu, ``{Decoding a
  Three-Dimensional Conformal Manifold},''
  \href{http://dx.doi.org/10.1007/JHEP02(2018)062}{{\em JHEP} {\bf 02} (2018)
  062}.

\bibitem{Keller:2012mr}
C.~A. Keller and H.~Ooguri, ``{Modular Constraints on Calabi-Yau
  Compactifications},'' \href{http://dx.doi.org/10.1007/s00220-013-1797-8}{{\em
  Commun. Math. Phys.} {\bf 324} (2013)  107--127},
  \href{http://arxiv.org/abs/1209.4649}{{\tt arXiv:1209.4649 [hep-th]}}.

\bibitem{Lin:2015wcg}
Y.-H. Lin, S.-H. Shao, D.~Simmons-Duffin, Y.~Wang, and X.~Yin, ``{$ \mathcal{N}
  $ = 4 superconformal bootstrap of the K3 CFT},''
  \href{http://dx.doi.org/10.1007/JHEP05(2017)126}{{\em JHEP} {\bf 05} (2017)
  126}, \href{http://arxiv.org/abs/1511.04065}{{\tt arXiv:1511.04065
  [hep-th]}}.

\bibitem{guillarmou2020}
C.~Guillarmou, A.~Kupiainen, R.~Rhodes, and V.~Vargas, ``{Conformal bootstrap
  in Liouville Theory},'' \href{http://arxiv.org/abs/2005.11530}{{\tt
  arXiv:2005.11530 [math.PR]}}.

\bibitem{xAct}
J.~M. Mart\'in-Garc\'ia, ``{xAct: Efficient tensor computer algebra for the
  Wolfram Language}.'' \url{http://www.xact.es}.

\end{thebibliography}\endgroup

\end{document}